\pdfoutput=1
\documentclass{article}
\usepackage{arxiv}
\usepackage[utf8]{inputenc}
\usepackage[style=numeric,citestyle=numeric-comp,sorting=none,backend=biber]{biblatex} 
\AtEveryBibitem{\clearfield{url}} 
\AtEveryBibitem{\clearfield{note}} 
\AtEveryBibitem{\clearfield{urldate}}
\renewbibmacro{urldate}{\iffieldundef{url}{}{\printurldate}}

\usepackage[table,xcdraw,dvipsnames]{xcolor}

\usepackage{tikz}
\usepackage{multirow}
\usetikzlibrary{patterns}
\usetikzlibrary{quantikz2}
\usepackage[colorlinks=true,citecolor=blue,linkcolor=magenta]{hyperref}
\usepackage{graphicx}
\usepackage{pgf}
\usepackage{caption}
\usepackage{subcaption}
\usepackage{dcolumn}
\usepackage{bm}
\usepackage{stmaryrd}
\usepackage[bbgreekl]{mathbbol}
\usepackage{amsmath,dsfont,mathtools}
\usepackage{amssymb}
\usepackage[normalem]{ulem}
\usepackage{color}
\usepackage{stmaryrd}
\usepackage{adjustbox}
\usepackage{physics}
\usepackage{amsfonts}
\usepackage{soul} 
\usepackage{todonotes}
\usepackage{overpic}
\usepackage{authblk}

\usepackage{amsthm}
\usepackage{comment}
\usepackage{tabularx}
\usepackage{algorithm}
\usepackage[noend]{algpseudocode}

\usepackage{csquotes}
\MakeOuterQuote{"}

\usepackage{thmtools}
\usepackage{thm-restate}
\usepackage{hyperref}
\makeatletter
\renewcommand{\thefootnote}{{\@fnsymbol\c@footnote}}
\makeatother

\definecolor{border_green_encoding}{HTML}{3C552D}
\definecolor{green_encoding}{HTML}{B8E986}
\definecolor{border_blue_ansatz}{HTML}{154239}
\definecolor{blue_ansatz}{HTML}{7FBFB2}
\definecolor{blue1}{HTML}{548ea1}
\definecolor{blue2}{HTML}{7f88bf}
\definecolor{blue3}{HTML}{7fbfb2}
\definecolor{grayplot}{HTML}{d6d6d6}

\newtheorem{definition}{Definition}

\definecolor{customblue}{HTML}{17178B}

\newcommand{\Id}{\mathds{1}}
\newcommand{\lpr}{\left(}
\newcommand{\rpr}{\right)}
\newcommand{\lbr}{\left[}
\newcommand{\rbr}{\right]}
\newcommand{\lb}{\left|}
\newcommand{\rb}{\right|}
\newcommand{\pr}[2][]{
	\mathop{
		\ifx &#1&
		\mathrm{Pr}
		\else
            \mathrm{Pr}_{#1}
		\fi
		\left[#2\right]}
}

\newcommand{\s}{\bm{s}}
\renewcommand{\t}{\bm{t}}
\newcommand{\p}{\bm{p}}

\newcommand{\per}[1]{\mathop{\mathsf{Per}\left(#1\right)}}
\usepackage{graphbox}
\newcommand{\ie}{{i.e.},\ }
\newcommand{\eg}{{e.g.},\ }

\newcommand{\Oc}{\mathcal{O}}
\newcommand{\bigo}[1]{\Oc(#1)}

\newcommand{\mmstate}[1]{\mathbb{1}\kern-1pt /\kern-1pt #1 }

\definecolor{quantumviolet}{HTML}{53257F}
\definecolor{navy}{RGB}{47,60,126}
\definecolor{darkviolet}{RGB}{99,56,142}
\definecolor{darkgreen}{RGB}{39,174,96}

\newcommand{\xmark}{%
\textcolor{red}{
\tikz[scale=0.23] {
    \draw[line width=0.7,line cap=round] (0,0) to [bend left=6] (1,1);
    \draw[line width=0.7,line cap=round] (0.2,0.95) to [bend right=3] (0.8,0.05);
}}}

\newcommand{\cmark}{%
\textcolor{darkgreen}{
\tikz[scale=0.23] {
    \draw[line width=0.7,line cap=round] (0.25,0) to [bend left=10] (1,1);
    \draw[line width=0.8,line cap=round] (0,0.35) to [bend right=1] (0.23,0);
}}}

\usepackage[nameinlink]{cleveref} 
\crefname{section}{Section}{Sections}
\crefname{equation}{Equation}{Equations}
\crefname{figure}{Figure}{Figures}
\crefname{table}{Table}{Tables}
\crefname{appendix}{Appendix}{Appendices}
\crefname{theorem}{Theorem}{Theorems}
\crefname{thm}{Theorem}{Theorems}
\crefname{cor}{Corollary}{Corollaries}
\crefname{lemma}{Lemma}{Lemmas}
\crefname{proposition}{Proposition}{Propositions}
\crefname{definition}{Definition}{Definitions}
\crefname{algorithm}{Algorithm}{Algorithms}

\let\autoref\cref

\hypersetup{
    colorlinks=true,
    linkcolor=navy,
    citecolor=quantumviolet,
    urlcolor=quantumviolet,
    breaklinks=true,
}

\title{Towards quantum advantage with photonic state injection}

\author[1,2]{Léo Monbroussou}
\author[1,$\dagger$]{Eliott Z. Mamon}
\author[1,3,4,$\dagger$]{Hugo Thomas}
\author[1,$\dagger$]{Verena Yacoub}

\author[4]{Ulysse Chabaud}
\author[1,5]{Elham Kashefi}

\affil[1]{Sorbonne Université, CNRS, LIP6, 75005 Paris, France}
\affil[2]{CEMIS, Direction Technique, Naval Group, 83190 Ollioules, France}
\affil[3]{Quandela, 7 rue Léonard de Vinci, 91300 Massy, France}
\affil[4]{DIENS, Ecole Normale Supérieure, PSL University, CNRS, INRIA, 45 rue d'Ulm, 75005 Paris, France}
\affil[5]{School of Informatics, University of Edinburgh, 10 Crichton Street, EH8 9AB Edinburgh, United Kingdom}
\affil[$\dagger$]{\emph{These authors contributed equally.}}

\begin{document}

\maketitle

\begin{abstract}
    We propose a new scheme for near-term photonic quantum device that allows to increase the expressive power of the quantum models beyond what linear optics can do. This scheme relies upon state injection, a measurement-based technique that can produce states that are more controllable, and solve learning tasks that are not believed to be tackled classically. We explain how circuits made of linear optical architectures separated by state injections are keen for experimental implementation. In addition, we give theoretical results on the evolution of the purity of the resulting states, and we discuss how it impacts the distinguishability of the circuit outputs. Finally, we study a computational subroutines of learning algorithms named probability estimation, and we show the state injection scheme we propose may offer a potential quantum advantage in a regime that can be more easily achieved that state-of-the-art adaptive techniques. Our analysis offers new possibilities for near-term advantage that require to tackle fewer experimental difficulties.
\end{abstract}
\section{Introduction}

Although quantum computers promise large advantages over classical computing, fault-tolerant universal quantum computers are still far from being available. In particular, while photonics is one of the promising platforms for quantum computing, the technological requirements for such photonic quantum devices are huge. They often rely on the capacity to achieve adaptive measurement-based operations \cite{knill_scheme_2001, knill_quantum_2002}, and to have access to a large number of modes and initial coherent photons. 

In the meantime, sub-universal models have been proposed to achieve near-term quantum advantage. Those models are believed to have an intermediate computational advantage even without being able to achieve every operation that a fault-tolerant quantum computer could do. Boson Sampling \cite{Gard_2015}, Gaussian Boson Sampling \cite{Zhong_2020}, or IQP circuits sampling \cite{Bremner2010ClassicalSO} are good candidates, but the range of problems that one can solve using such approaches seems very limited. Finding an architecture able to offer a quantum utility to real life use case for quantum photonic device in the era of Noisy Intermediate-Scale Quantum \cite{Preskill_2018} devices is an important field of research. Previous work also used linear optic circuits with post-processing strategy to simulate universal quantum computing \cite{polino2022photonicimplementationquantumgravity}, at the cost of strongly increasing the running time of the produced algorithms.  

\begin{figure}[h]
    \centering
    \includegraphics[width=0.95\textwidth]{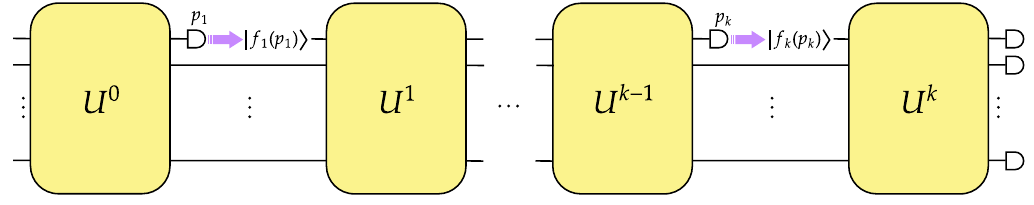}
    \caption{Quantum circuit made of linear optical blocks separated by state injections.}
    \label{fig:Introduction_framework}
\end{figure}

Quantum Machine Learning (QML) is a promising field of application for photonic platforms that requires to go beyond linear optic architectures to increase the expressivity of the resulting models, but that does not necessarily depend on whether the model is universal for quantum computation. Previous works came up with alternative solutions to incorporate non-linearity --- an important element for neural network architectures --- such as global measurements and classical activation functions between linear optical layers \cite{steinbrecher2018}, physical non-linear blocs \cite{fu2023photonic}, or adaptive gates to perform learning tasks \cite{chabaud_quantum_2021}.

In this work, we study the use of linear optical layers interleaved with state injection (SI), a measurement-based channel that does not require gate adaptivity throughout the computation, but rather the preparation and injection in the circuit of new quantum optical state. We give a detailed description of this non-linear gadget in \autoref{sec:State_Injection}. In particular, we explain why it can be easier to implement on a photonic platform with respect to \eg the proposal of \cite{chabaud_quantum_2021}. In \autoref{sec:Expressivity}, we offer a theoretical perspective on the lack of expressivity of linear optical circuits and explain how SI can increase the controllability of the output state, an important figure of merit of quantum circuit expressivity. In addition, we study how non-linear layers affect the model outputs, leading to theoretical results of separate interest on the connection between the purity of the states and the distinguishability of the output models. Finally, we offer in \autoref{sec:Proba_Estimation} an instance of a problem that our framework can solve with a potential quantum advantage in comparison with classical algorithms, and that can be implemented more easily than state of the art methods.  
\section{State Injection}\label{sec:State_Injection}

In this Section, we introduce the State Injection (SI) scheme. This method allows to increase the expressivity of the photonic quantum circuits as explained in \autoref{sec:Expressivity}, and to perform tasks that are believed hard to do classically with fewer experimental constraints in comparison with the Adaptive Linear Optics (ALO) scheme~\cite{chabaud_quantum_2021}, as explained in \autoref{sec:Proba_Estimation}. In the following, we define the SI scheme in \autoref{subsec:SIdef} and we motivate its experimental implementation in \autoref{subsec:general_exp_framework}.

    \subsection{Scheme Defintion}\label{subsec:SIdef}

The ALO scheme, that we illustrate in \autoref{fig:Experimental_Comparison_AdaptiveFeedForward}, was proposed by~\cite{chabaud_quantum_2021} where the authors built a \emph{feed-forward} scheme for linear optical quantum computation. Their setup is composed of an input Fock state, with $n$ photons spread across $m$ modes, and $k$ adaptive measurements. 
An adaptive measurement consists in measuring one mode using photon number resolving detector and configure the following $(m-1)\times(m-1)
$ unitary according to the measurement result.

We illustrate the SI scheme in \autoref{fig:Introduction_framework}. It differs from the ALO scheme by the following. Firstly, no real time reconfiguration of the quantum circuit is required to be done within the run of the experiment as all unitary matrices are preset for each run. Secondly, our adaptive part comes from choosing what is the new Fock state to be injected in the circuit after a measurement. Lastly, the unitary does not shrink in size as the number of modes is preserved throughout the computation.

\begin{definition}[State Injection]\label{def:StateInjection}
    We call \textbf{State Injection} (SI) any operation on an $m$-mode photonic platform that performs photon-counting measurements in one or several modes, and, depending on the outcomes obtained, re-injects some photons back in one or several modes.
    Overall (since no single outcome is post-selected on), this operation is described by a CPTP map on the $m$-mode, $n$-photon Hilbert space.
    Different SI operations hence correspond to different choices of modes that undergo measurements and/or re-injections, and different choices of rules that map measurement outcomes to the corresponding re-injections that should be performed. We refer to the latter as a choice of \textbf{injection functions}.
\end{definition}

In the special case where the same $k$ modes (here written as adjacent modes for simplicity) are being subject to both measurements and re-injections, and the total photon count is preserved, SI operations may be detailed as follows. 
If the photon count outcomes obtained when measuring those modes are $n_1,\dots,n_k$, then the re-injection process consists of injecting, \emph{in those modes only}, the new state $\ket{f_1(n_1,\dots,n_k),\cdots,f_k(n_1,\dots,n_k)}$, where the injection functions here take the form of maps $f_i:[\![0,n]\!]^k \to [\![0,n]\!]^k$ that may be chosen arbitrarily among those that respect the photon-number conservation constraint
\begin{equation}\label{eq:StateInjection-SpecialCase-photon-conservation-constraint}
    f_1(n_1,\dots,n_k) + \cdots + f_k(n_1,\dots,n_k) = n_1 + \cdots + n_k\,.
\end{equation}

SI is a new tool that allows to go beyond the standard linear quantum optical circuits such as those used in Boson Sampling schemes. The infinite possibilities of injection functions and measurements offer new possibilities, but we will mainly focus, in this paper, on the case where we count the photons in one mode and inject the same number of photons. This choice is motivated by experimental consideration, and by the preservation of the subspaces defined by a fixed number of particles. We use this example in \autoref{subsec:Purity}, and in \autoref{sec:Proba_Estimation}. In  \autoref{subsec:general_exp_framework} we start by analyzing the general experimental framework that is common for both scheme, building on that we will highlight the differences between the ALO and SI schemes.

\begin{figure}[h!t]
\centering
\begin{subfigure}{.5\textwidth}
    \centering
    \includegraphics[align=c,width=1.05\linewidth]{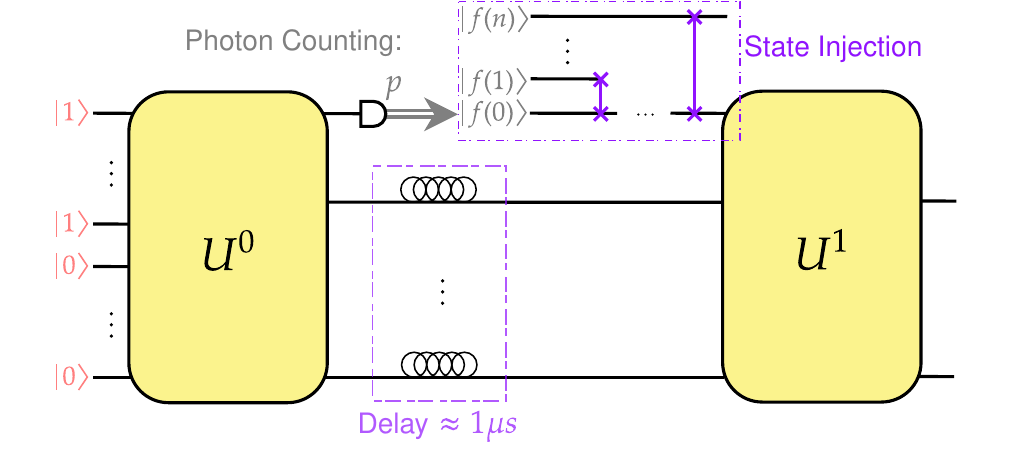}
    \caption{State Injection channel for a single mode measured.}
    \label{fig:Experimental_Comparison_StateInjection}
\end{subfigure}%
\hspace*{.1in}
\begin{subfigure}{.49\textwidth}
      \centering
    \includegraphics[align=c,width=1.05\linewidth]{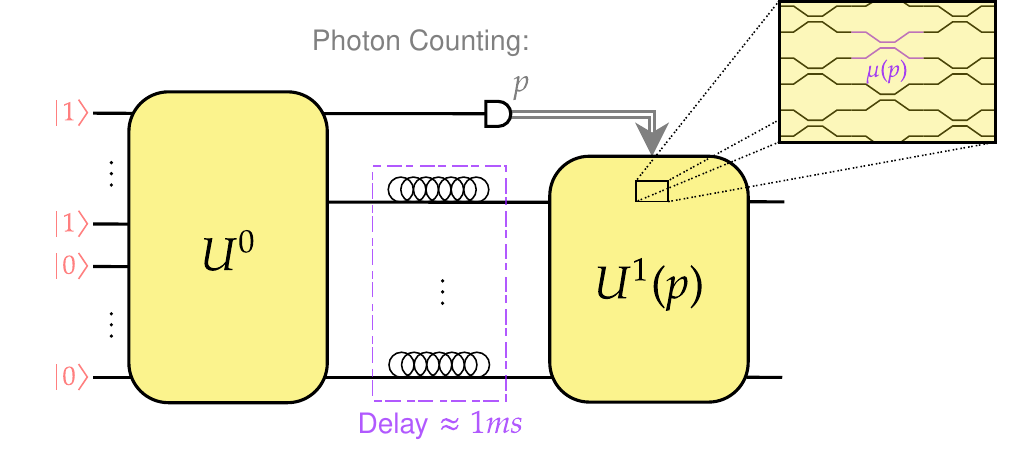}
    \caption{Feed-Forward adaptivity for a single mode measured.}
    \vfill
    \label{fig:Experimental_Comparison_AdaptiveFeedForward}
\end{subfigure}
\caption{Comparison of experimental requirements for the State Injection scheme (in \autoref{fig:Experimental_Comparison_StateInjection}) that we introduce, and the Feed-Forward scheme (in \autoref{fig:Experimental_Comparison_AdaptiveFeedForward}) proposed in \cite{chabaud_quantum_2021}.}
\label{fig:Experimental_Comparison}
\end{figure}

It is important to keep in mind that the challenges of any scheme depend on the exact task and corresponding quantum circuit. In \autoref{sec:Proba_Estimation}, we give an instance of a learning problem where the SI scheme can require less resources. In a more "near-term" perspective, one can adapt the schemes to implement a sub-universal model that matches the experimental capacity of a platform. For example, one can choose injection functions or the adaptive function in the feed-forward scheme to not depend on the measured number of photons to avoid photon counting.

\subsection{General experimental framework} \label{subsec:general_exp_framework}

In this part, we discuss the general experimental framework needed for both ALO and SI schemes. We can see that both proposals rely on a general scheme consisting of: State preparation, unitary, measurement, and adding a new unitary. In the following, we will expand the general experimental requirements for each step: 

\textbf{State Preparation.} Both schemes start with multiple single photons Fock states which requires a photon source that could emit indistinguishable photons on many parallel modes synchronously. The most common options for this are Quantum dots (QD) and time multiplexing~\cite{maring2024versatile,wang2017high} or Spontaneous Parametric Down Conversion (SPDC) sources~\cite{zhong201812}.  

\textbf{Unitary.} The unitaries should be fully programmable which is the case for the photonic chips. Photonic chips, or processors, are made of waveguides shaped in beam splitters, the reflectivities of the latters are controlled by phase shifters. The waveguides are mostly engraved in glass or in silicon nitride. The phase shifters, which are the programmable parts,could be controlled by thermo-optical effects which can be quite slow taking from hundreds of microseconds to milliseconds for each reconfiguration~\cite{calvarese2022strategies,taballione202320} or piezo-electrical (optomechanical) effects that are usually faster, tuned in hundreds of microseconds, but they have challenges for scalability~\cite{tian2024piezoelectric}.    

\textbf{Measurements.} Both schemes rely on single photon measurements between the unitaries. ALO requires Photons number resolving (i.e. photon counting) while SI can accept threshold detectors as well (i.e. detecting only the presence and the absence of photons) depending on the application.  

\textbf{Addition of unitary.} To be able to use the single photon measurements in real time, one should mostly consider independent chips, which implies that the spatial modes of both chips would be linked with optical fibers which imposes the implementation of a good temporal synchronization of all the modes. Noting that when we mention synchronization between the photons in state preparation or the modes between the unitaries we mean matching the temporal delays in order to maximize Hong-Ou-Mandel effect~\cite{hong1987measurement}.

In the case of the ALO scheme, the main challenge that needs to be adressed is the waiting time to reprogram the chips after each measurements. For example, if we consider an implementation with the most common setup relying on thermo-optical effects, we will have to put delays (or quantum memories) between the unitaries that are in the order of milliseconds, equivalent to few hundreds of kilometers of fibers which will have more losses and harder synchronization.

In the new SI scheme, we relieved the need of the real time adaptability of the consequent unitaries, we choose preset parameters for all of them instead. We rely on active real-time Fock state preparation to replace the few measured modes depending on their measurement outcome and synchronizing the new states with the unmeasured modes of each unitary, using mainly external (to the unitary) fast optical switches which would reduce the delay required for the non measured modes. Average commercial switches relying on micro or nano electro-optical technologies~\cite{thomas2024noise, memeo2024micro} can work in microseconds range which reduces the optical delay to the order of few hundreds meters.
For some applications, as in \autoref{sec:Proba_Estimation}, this scheme will also allow the reduction of the number of photons needed at the beginning of the experiment.

Notice that a similar scheme has been recently proposed in~\cite{sulimany2024quantum}, where they partially measure a coherent pulse by heterodyne on some modes and accordingly apply an adaptive displacement operator on the re-injected coherent state, but their purpose from that was to have a notion of security rather than an advantage in computing.

\section{Improving the controllability of photonic circuits}\label{sec:Expressivity}

In this Section, we offer theoretical arguments that SI can improve the quantum model produced by a circuit made of linear optic layers separated by state injections (see \autoref{fig:Introduction_framework}). After recalling the framework of ordinary linear quantum optical circuits in \autoref{subsec:Structure_BS}, we highlight in \autoref{subsec:Limitation_BS} their inherent limitation of controllability, defined in terms of degrees of freedom that the parameters may explore in state space. Then, we explain in \autoref{subsec:Controllability_Improvement} how to improve the controllability while protecting some interesting properties of linear optic circuits. Finally, we study in \autoref{subsec:Purity} how SI, while providing an increase in control, reduces the purity of the output state. We provide there results on the purity evolution, and on the relation between low purity of output states and their low distinguishability.

\subsection{Structure of linear quantum optics}\label{subsec:Structure_BS}

We consider linear-optical network with $m$ modes and set of simple optical elements (beam-splitters and phase-shifters). In general, those circuit are used while considering input states made with $n$ identical photons that pass through the modes and optical elements and then measured to determined their locations. Here, we do not consider any adaptive photon-number measurements. 

We consider linear-optical networks/circuits made of simple optical elements, beam-splitters and phase-shifters, over $m$ modes. Each element may either be regarded as \textit{fixed}, or \textit{parameterized}, meaning that the gate has a tunable parameter $\theta \in [0,2\pi]$ that can be freely varied --- corresponding to a beam-splitter's \textit{angle} or a phase-shifter's \textit{phase}. If there are $p$ parameterized gates in the circuit, we denote by $\bm{\theta} \in \Theta:=[0,2\pi]^p$ the tuple of all the parameter values. Into the circuit are sent $n$ (indistinguishable) photons in some pure quantum state. A pure quantum state of $n$ photons is a normalized vector in the \textit{$n$-photon Fock space}, which is the Hilbert space of all complex superpositions of the \textit{basis Fock states} $\ket{\bm s}$ (for all $\bm s = (s_1\dots,s_m)\in \mathbb N^m$ such that $s_1 + \dots + s_m = n$). We denote the set of basis Fock states by $\Phi_{m,n}$, of which there are $d_n := |\Phi_{m,n}| = \binom{m+n+1}{n}$ many. The direct sum of all the $n$-photon Fock spaces is known as \textit{the Fock space}, which is infinite-dimensional. We will use the notations $|\bm s| = \sum_{i=1}^m s_i$ and $\bm s! = \prod_{i=1}^m s_i !$. 

An arrangement of beam-splitters and phase-shifters over $m$ modes specifies a given unitary $W^1 \in SU(m)$\footnote{For simplicity, in this work we only ever consider \textit{special}-unitary matrices, i.e. unitary matrices $U$ with $\det(U)=1$. This is without loss of generaility, as it just amounts to a convention choice in how one writes the $m \times m$ unitary matrices representing beam-splitters and phase-shifters on the $m$-mode system.}, which dictates the evolution of a single photon in the circuit. Conversely, all unitaries in $SU(m)$ may be realized as some arrangement of beam-splitters and phase-shifters \cite{Reck-ExperimentalRealization-1994}. A single-photon unitary $W^1 \in SU(m)$ determines an $n$-photon
unitary $W^n \in SU(d_n)$, through the so-called \textit{($n$-photon) photonic homomorphism}: $W^n := \varphi(W^1)$ (see \autoref{thm:BS_Model_Bloc}). Since quantum linear optics preserves the photon number $n$, the global unitary $W$ (over the whole Fock space) of the circuit corresponds to an inifitly-sized block-diagonal matrix, where the blocks are the evolutions for a fixed photon number. Lastly, a fixed architecture of beam-splitters and phase-shifters will give rise to a parametrization $\bm{\theta}\mapsto W^1(\bm{\theta})$, and so the spaces of accessible unitaries (for all possible parameter values $\bm{\theta}$) will generally be smaller. We summarize the different unitary matrices introduced, the spaces they live in, and the homomorphism relation $\varphi$, in \autoref{fig:Subspace}.
\begin{figure}[h!t]
\centering
\begin{subfigure}{.5\textwidth}
    \centering
    \includegraphics[align=c,width=1.0\linewidth]{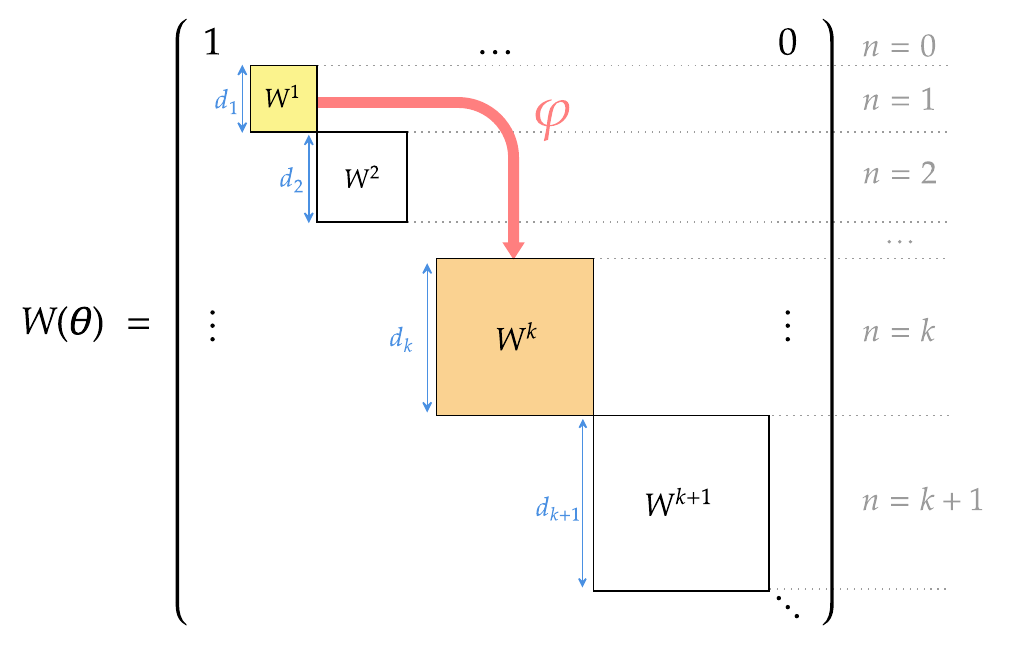}
    \caption{Representation of linear optic quantum circuit equivalent unitary matrix as a bloc diagonal matrix. Each bloc $W^k$ corresponds to a subspace of $k$ particles.}
    \label{fig:Unitary_Bosonic_Circuits}
\end{subfigure}%
\hspace*{.1in}
\begin{subfigure}{.49\textwidth}
      \centering
    \includegraphics[align=c,width=1.0\linewidth]{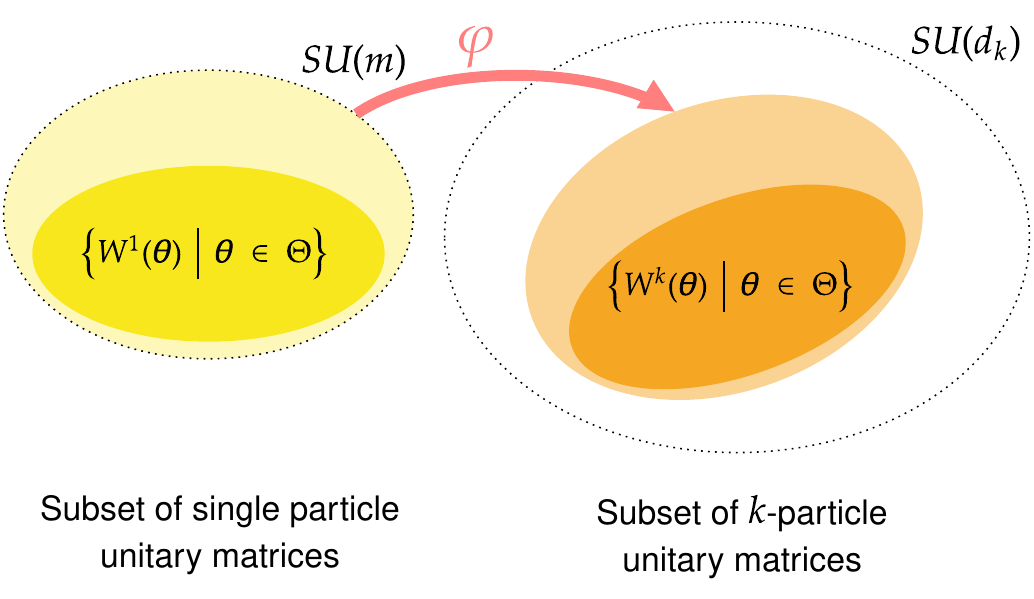}
    \caption{Representation of the spaces of reachable unitaries using linear quantum optics over $m$ modes, for $1$ photon (left) and $k$ photons (right). These spaces are shown in translucent color, while the opaque colors are the subsets that arise through a specific parametrization $\bm{\theta}\mapsto W^1(\bm{\theta})$ in terms of beam-splitters and phase-shifters.}
    \vfill
    \label{fig:subspace_representation}
\end{subfigure}
\caption{Representation of the particle number preserving equivalent unitary matrix (\autoref{fig:Unitary_Bosonic_Circuits}) and illustration of the manifold described by the achievable set of matrices for different number of particles (\autoref{fig:subspace_representation}). The homomorphism $\varphi$ is described in \autoref{thm:BS_Model_Bloc}. Each bloc $W^n$ is a parametrized unitary matrix of dimension $d_n = \binom{m+n+1}{n}$.}
\label{fig:Subspace}
\end{figure}

The homomorphism $\varphi:SU(m)\to SU(d_n)$ describes the way in which second-quantization enforces the evolution of $n$ indistinguishable bosons, given the evolution unitary for a single boson. There are different equivalent ways to describe $\varphi$. One approach builds an expression for $\varphi(W^1)$ in terms of the matrix \textit{permanents} \cite{marcus_permanents_1965} of certain matrices related to $W^1$ \cite{aaronson_computational_2011}; another approach is to leverage the fact that $\varphi$ is an injective homomorphism between two Lie groups, and so can be understood in terms of its derivative action on the Lie Algebra $\mathfrak{su}(m)$ \cite{Parellada-NogoTheorems-2023}.
We recall in \autoref{thm:BS_Model_Bloc} the first approach's expression.


\begin{restatable}[{Photonic homomorphism in terms of matrix permanents, from \cite[Section 3]{aaronson_computational_2011}}]{thm}{BSModel}
\label{thm:BS_Model_Bloc}
Given a unitary $W^1(\bm{\theta}) \in SU(m)$ describing an $m$-mode linear optical circuit (possibly parametrized by $\bm{\theta}$), the corresponding unitary $W^n(\bm{\theta}):=\varphi(W^1 (\bm{\theta}))$ describing the $n$-photon evolution is given, for all $\bm s, \bm t \in \Phi_{m,n}$, by
\begin{equation}\label{eq:bloc_unitary_expression_BS}
    \bra{\bm s} W^n(\bm{\theta}) \ket{\bm t} = \frac{\per{I_S^\dagger W^1(\bm{\theta}) I_T}}{\sqrt{\bm s! \bm t !}} \, ,
\end{equation}
with $I_S$ a linear substitution of variables such that $I_S [x_1, \dots, x_m] = [x_{s_1}, \dots, x_{s_m}]$.
\end{restatable}

In the following we will refer to the \textit{dimension} of various subsets of unitary matrices, or of density matrices. We precise what mean by the word dimension in \autoref{chap:Details_Dimension}, but in short, it corresponds to a \emph{local manifold dimension}.

\subsection{Limited controllability of linear quantum optics}\label{subsec:Limitation_BS}

Consider a parametrized linear optical circuit, $W^1(\bm{\theta})$, with $p$ parametrized gates, \ie $\bm{\theta} \in \Theta=[0,2\pi]^p$. The space of $m \times m$ unitaries that are accessible, as all parameters are explored, is by definition included in $SU(m)$. Therefore (see \autoref{chap:Details_Dimension}),
\begin{align}
    \quad \dim(\{ W^1(\bm{\theta}) \; | \; \bm{\theta} \in \Theta\}) \leq \min(p,\ m^2 - 1).
\end{align}


For the corresponding $n$-photon unitaries, the existence of the injective homomorphism $\varphi$ implies that the dimension of the set of the $n$-photon unitaries reached is equal to that of the single-photon unitaries, and therefore it obeys the same limitations:
\begin{equation}
    \dim(\{ W^n(\bm{\theta}) \; | \; \bm{\theta} \in \Theta\}) = \dim(\{ W^1(\bm{\theta}) \; | \; \bm{\theta} \in \Theta\}) \leq \min(p,\ m^2 - 1).
\end{equation}

This limitation in the set of achievable unitary matrices dimension is a constraint on the expressivity of the model output. Other figures of merit for the expressivity of quantum models exist, including the distance to a 2-design \cite{PRXQuantum.3.010313} that characterizes the distribution of the unitary matrices, or the Fourier expressivity \cite{xiong2023, mhiri2024}. In this work, we focus on the notion of \textit{controllability} of the output state of the quantum circuit, \ie the number of independent directions it can locally explore in the space of density matrices.

Recent works have highlighted the impact of the controllability of quantum circuits, especially in QML, by \eg studying the rank of the Quantum Fisher Information Matrix (QFIM) \cite{Larocca_2023}, or the dimension of the Dynamical Lie Algebra \cite{Larocca_2022, Ragone2023, Fontana2023} generated by the Hamiltonians in the circuit. Since state injections are not unitary evolutions (if no post-selection on a given outcome is done) but CPTP maps, we study the controllability at the level of the output state's density matrix.  

Studying SI forces us to consider a tool to characterize the controllability that is not only defined for unitary transformations. Accordingly, we introduce a new measure of controllability of the output state of a parametrized quantum circuit, using its corresponding Jacobian rank.

\begin{definition}[Number of degrees of freedom of a state]\label{def:DoF}
    The \textbf{number of degrees of freedom} of an $n$-photon state $\rho(\bm{\theta})$ at a point $\bm{\theta}$ in the parameter space is defined as the rank of Jacobian matrix of the map $\rho:\Theta\to\mathbb{C}^{d_n \times d_n}$ calculated at point $\theta$:
    \begin{equation}
        \mathrm{DoF}(\rho(\bm{\theta})) = \rank[J \rho(\bm{\theta})]\,.
    \end{equation}
    The Jacobian matrix considered is actually the one of the map $\tilde{\rho}:\Theta\to \mathbb{R}^{2d_n^2}$ that results from viewing a complex matrix $\rho(\theta) \in \mathbb{C}^{d_n \times d_n}$ as a real vector $\tilde{\rho}(\theta)$ of length $2d_n^2$ (through concatenation of all the columns of the matrix, and splitting of each complex scalar into its real and imaginary parts). We recall that the Jacobian matrix of a differentiable function $f:\mathbb{R}^a \to \mathbb{R}^b$ at point $\bm{x}\in\mathbb{R}^a$ is the $b \times a$ real matrix defined by $\big[J f (\bm{x})\big]_{ij}:=\frac{\partial f_i}{\partial x_j}$.
\end{definition}

Since the considered circuits may apply successively operations that are standard parameterized optical gates, fixed optical gates, and state injections, overall the map $\bm{\theta} \mapsto \rho(\bm{\theta})$ that sends a parameter tuple to the final state's density matrix is \emph{analytic}\footnote{A function is \textit{analytic} if it is smooth and if it agrees locally with its Taylor series around each point in the domain. Sines, cosines, as well as matrix products, sums, and exponentials, are all analytic; and compositions of analytic functions are still analytic. Hence, all the maps $\bm{\theta} \mapsto \rho_{\mathrm{out}}(\bm{\theta})$ considered in this work --- even those including state injections --- are analytic, being only compositions of sine, cosines, and products and sums of matrices.}, and consequently, the number of degrees of freedom $\bm{\theta} \mapsto \mathrm{DoF}(\rho(\bm{\theta}))$ is constant \textit{almost-everywhere} on the parameter space, due to \cite[Prop. B.4]{Bamber-HowMany-1985} (a more precise writing would be exactly the same as \cite[Lemma 4]{monbroussou2023}, with rank of density matrix's Jacobian in place of rank of the pure state's Quantum Fisher Information Matrix (QFIM); see also \autoref{chap:Details_Dimension}).

\begin{restatable}[Almost-constant property of number of degrees of freedom]{thm}{QFIMrankthm}
\label{thm:DoFthm}
Let $\rho(\bm{\theta})$ be the density matrix of the output state of a linear optical circuit, with or without state injections. Then, its number of degrees of freedom is, \emph{almost everywhere} on the considered parameter space $\Theta$, constant and equal to
\begin{equation}\label{eq:DofThmEq}
    \mathrm{DoF}_{\mathrm{max}}(\rho):=\max\limits_{\bm{\theta} \in \Theta}\mathrm{DoF}(\rho(\bm{\theta})).
\end{equation}

\end{restatable}
The practical consequence of \autoref{thm:DoFthm} is that drawing a point $\bm{\theta} \in \Theta$ uniformly at random and calculating the state's number of degrees of freedom at that point yields $\mathrm{DoF}(\rho(\bm{\theta}))=\mathrm{DoF}_{\mathrm{max}}(\rho)$ with probability $1$.

Hence in this work, we numerically evaluate the quantity $\mathrm{DoF}_{\mathrm{max}}(\rho)$ as follows: given an optical circuit over $m$ modes and an $n$-photon input state $\rho_{\mathrm{in}}$, we classically simulate the $d_n \times d_n$ output density matrix $\rho_{\mathrm{out}}(\bm{\theta})$ through successive applications of beam-splitters and phase-shifters unitary channels and state injection CPTP maps, using Python's library \textit{PyTorch} \cite{PyTorch}, which, by relying on automatic differentiation, enables us to access the Jacobian $J \rho_{\mathrm{out}}(\bm{\theta})$ of the output state. Then, we draw uniformly at random $\bm{\theta}$, and calculate $\rank[J \rho(\bm{\theta})]$. With probability 1, this number is equal to $\mathrm{DoF}_{\mathrm{max}}(\rho)$ (\autoref{eq:DofThmEq}).

Since the controllability of the output state $\rho_{\mathrm{out}}(\bm{\theta})$ lies entirely in the controllability of the single-photon unitary $W^1(\bm{\theta})$ of the whole circuit (see \autoref{chap:Details_Dimension}), we immediately have the following limitation.

\begin{restatable}[Controllability limitation of linear quantum optics]{thm}{ControlLimitationBS}
\label{thm:Control_Limits_BS}
    Consider an $n$-photon pure state ${\rho_{\mathrm{in}} := \ketbra{\psi_{in}}}$ entering an $m$-mode linear optical circuit $W^1(\bm{\theta})$ with $p$ parametrized gates, and without state injections.
    Then, the controllability of the output density matrix $\rho_{\mathrm{out}}(\bm{\theta}) := W^n(\bm{\theta}) \rho_{\mathrm{in}} W^n(\bm{\theta})^\dagger$ is bounded by
    \begin{equation}
        \mathrm{DoF}_{\mathrm{max}}(\rho_{\mathrm{out}}) \leq \dim(\{ W^n(\bm{\theta}) \; | \; \bm{\theta} \in \Theta\}) \leq m^2 - 1.
    \end{equation}
\end{restatable}
In fact, note that in the case where the input state consists of all photons in one same mode, i.e., ${\rho_{\mathrm{in}} = \ketbra{n,0,\dots,0}}$, an even tighter bound of $O(m)$ instead of $O(m^2)$ may be shown to hold.\footnote{This bound can be obtained by an explicit calculation of the rank of the action of the derivative of the photonic homomorphism onto this initial state.}

There exists subspace preserving quantum circuits that do not suffer from such controllability limitations, \eg RBS-based Hamming Weight preserving quantum circuits \cite{monbroussou2023}. Those ansatz are particularly useful as they are likely to avoid vanishing gradient phenomena, the so-called \emph{barren plateau} \cite{McClean2018}, when considering subspaces of polynomial dimension with respect to the number of qubits \cite{Larocca_2022, Ragone2023, Fontana2023, diaz2023}. They can often be classically simulated \cite{cerezo2024does, goh2023}, meaning that such quantum circuits would offer only polynomial advantage, if they are to offer any advantage at all. In the following, we explain how SI can increase the controllability of the quantum circuit while maintaining -- if needed -- the subspace preserving properties of the circuit.

\subsection{Controllability improvement with state injection}\label{subsec:Controllability_Improvement}

Since circuits that include the operations of state injection proposed in this work go beyond  linear optical circuits, they are not subject to \autoref{thm:Control_Limits_BS}, and therefore their output states $\rho_{\mathrm{out}}(\bm{\theta})$ may a priori enjoy controllability $\mathrm{DoF}(\rho_{\mathrm{out}})$ that goes beyond the $m^2 -1$ upper bound.

In \autoref{fig:DoF_Evolv}, we study the controllability of an example circuit with state injections, for $m=6$ modes, and $n=3$ photons, all input into the first mode. The type of state injection chosen here is the simplest one can consider: it is the special case of \autoref{eq:StateInjection-SpecialCase-photon-conservation-constraint} in \autoref{def:StateInjection} with $k=1$ and $f_1(n_1):=n_1$.
For both a circuit with these state injections, and the same one without the state injections (\autoref{fig:DoF_Comparison_Circuits}), we evaluate numerically the number of degrees of freedom $\mathrm{DoF}(\rho^{(i)})$ (\autoref{def:DoF}) of each intermediate state $\rho^{(i)}(\bm{\theta}^{(i)})$ of the circuit (\autoref{fig:DoF_Plot}). See \autoref{subsec:Limitation_BS} for more details about how these numbers are calculated.

\begin{figure}[h!t]
\centering
\begin{subfigure}{.5\textwidth}
    \centering
    \includegraphics[align=c,width=1.0\linewidth]{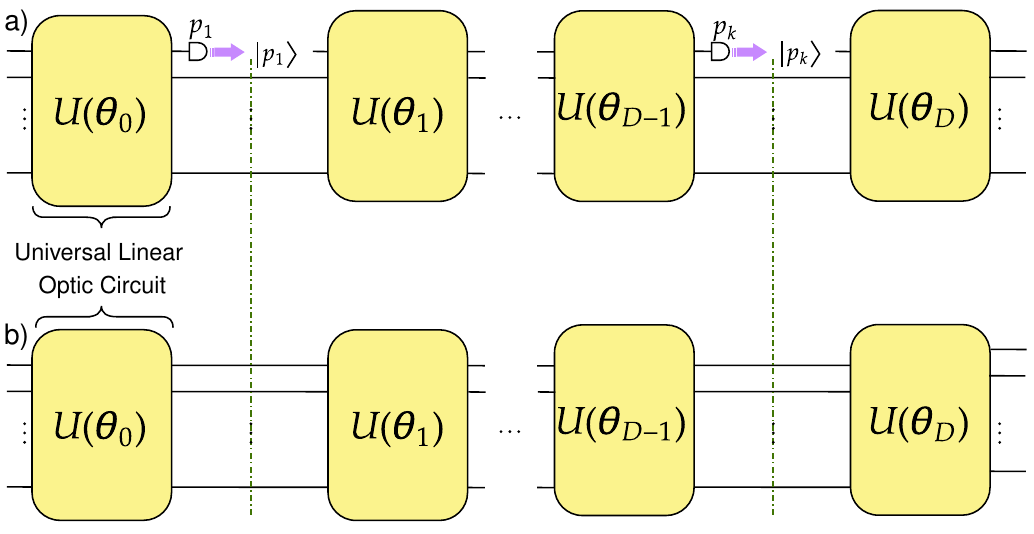}
    \caption{Both circuits are made of 6-mode Linear Optic block that are universal but in a) those blocks are separated with state injections, while in b) those blocks are connected.}
    \label{fig:DoF_Comparison_Circuits}
\end{subfigure}%
\hspace*{.1in}
\begin{subfigure}{.49\textwidth}
      \centering
    \includegraphics[align=c,width=1.1\linewidth]{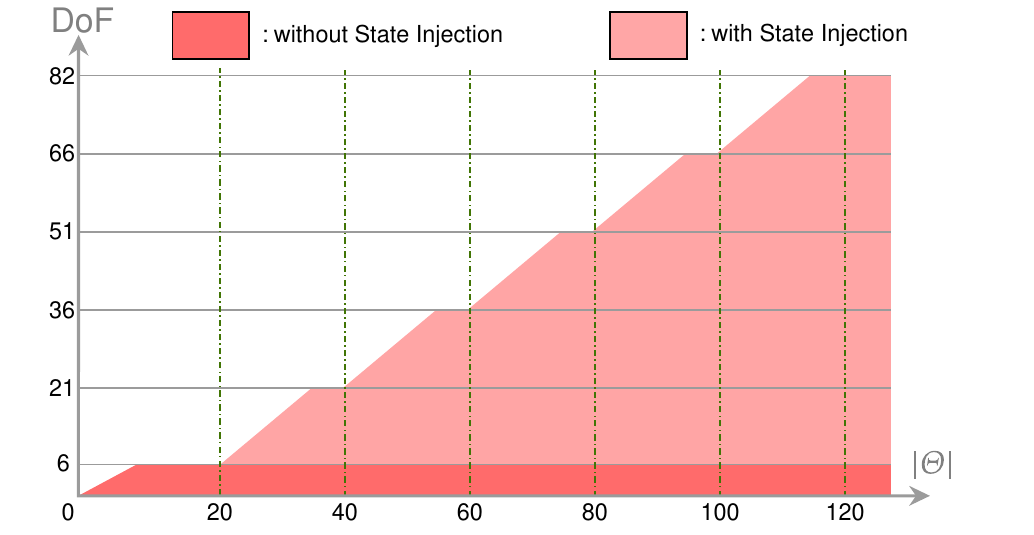}
    \caption{Evolution of the number of Degrees of Freedom of the state in the subspace of $3$ particles according to the number of beam-splitters we consider.}
    \vfill
    \label{fig:DoF_Plot}
\end{subfigure}
\caption{Degrees of Freedom (DoF) evolution comparison between two circuits. Each block $U(\bm{\theta}_i)$ is universal in the sense that it can reach any $6 \times 6$ orthogonal matrix in the subspace of $1$ particle. The vertical bars represent the application of the state injection operation, where we measure on one mode $p$ particles and we inject  on the measured mode the Fock state $\ket{p}$.}
\label{fig:DoF_Evolv}
\end{figure}

The unitaries in the circuit considered only consist of beam-splitters, spreaded over the $m=6$ modes. We consider each yellow block of beam-splitters in \autoref{fig:DoF_Comparison_Circuits} to be universal, which only requires $\dim(O(m)) = m(m-1)/2 = 15$ parameters, and we add 5 extra random beam-splitters to that to illustrate that they do not improve the controllability of the state beyond that.

We observe in \autoref{fig:DoF_Plot} that indeed, the use of state injections (green lines) has enabled breaking the limitation of controllability of output states of plain beam-splitters (see each little plateau just before each green line).

    \subsection{Purity evolution with state injections}\label{subsec:Purity}

Using non-unitary channels such as state injection, or other measurement-based methods, allows one to increase the controllability of the final state but decreases its purity $\Tr[\rho_{\mathrm{out}}^2]$. Many algorithms relying on non-linear channels \cite{coyle2024trainingefficientdensityquantummachine, Cong_2019} to increase the expressivity of their model do not consider the cost of reducing the purity of the final state. As controllability can be increased via SI, care should be taken not to reach the maximally mixed state. We therefore need to address the two following questions.
\begin{itemize}
    \item At what rate does the purity decrease when using SI ?
    \item What is the role of purity in that model ? 
\end{itemize}

For the sake of simplicity, we consider the special case in which the SI layer merely consists in measuring a single mode occupancy. 
In this case, we can state the following \autoref{thm:PurityEvolution}.

\begin{restatable}{thm}{PurityEvolution}\label{thm:PurityEvolution}
    Consider a quantum circuit made of $m$ modes with an initial pure state with $n$ photons. If there is a single SI layer in the circuit, which consists in measuring the number of photons $p \in [\![ 0, n ]\!]$ in one of the modes and re-injecting the state $\ket{p}$ in that same mode, then the purity of the final state is given by
    \begin{equation}
        \gamma(\rho_{\mathrm{out}}) = \Tr[\rho_{\mathrm{out}}^2] = \sum_{i=0}^n \Pr[i]^2\, \textrm{,}
    \end{equation}
    with $\Pr[i]$ the probability of measuring $i$ photons on intermediate state just preceding the state injection. 
\end{restatable}

The result of this theorem also stands for any injection function such that two different measurements implies the injection of orthogonal states. The proof is presented in \autoref{subchap:ProofTheoremPurityEvolution}. From this result, we derive in \autoref{cor:PurityEvolutionBound} a lower bound of the purity of a linear optical circuit based on the number of SI layers it contains.

\begin{restatable}{cor}{PurityEvolutionBound}\label{cor:PurityEvolutionBound}
     We consider a quantum circuit made of $m$ modes with an initial pure state with $n$ photons, and with $L$ SI layers separated by linear optical blocs. If each SI layer consists in measuring the number of photons $p \in [\![ 0, n ]\!]$ in one of the modes and re-injecting the state $\ket{p}$ in that same mode, then the purity of the final state is such that:
     \begin{equation}
         \gamma(\rho_{\mathrm{out}}) \geq \frac{1}{(n+1)^L}\, \textrm{.}
     \end{equation}
\end{restatable}

The proof is presented in \autoref{subchap:ProofCorWorstScenario}. In \autoref{cor:PurityEvolutionBound}, we consider the worst case scenario where each measurement outcome is as likely, resulting in this inverse-exponential lower bound in the purity. But if one has prior knowledge about which outcomes are more likely to occur, this lower bound may be tightened. In particular, when considering a number of modes much greater than the number of photons, one enters the so-called \textit{no-collision regime}, where the probability of measuring more than one photon is negligible. This phenomenon may be quantified using the Boson Birthday Bound introduced in \cite{aaronson_computational_2011}, and doing so, we get the following.

\begin{restatable}{cor}{PurityNoCollision}\label{cor:PurityNoCollision}
     We consider a quantum circuit made of $m$ modes with an initial pure state with $n$ photons, and with $L$ SI layers separated by linear optical blocs. We again consider that each SI layer consists in measuring the number of photons $p \in [\![ 0, n ]\!]$ in one of the modes and re-injecting the state $\ket{p}$ in that same mode.
     
     If $m > 2n^2$, and if the linear optical blocks are considered to each be Haar distributed in the single-photon subspace, then the purity of the final state is such that:
     \begin{equation}
         \mathbb{E}_{U \in \mathcal{H}_{m,m}} [\gamma(\rho_{\mathrm{out}})] \geq \left(\frac{m-2n^2}{\sqrt{2}m}\right)^{2 L}\,.
     \end{equation}
\end{restatable}

The proof is presented in \autoref{subchap:ProofCorNoCollision}. In this setting, we notice that the purity lower bound decreases at a lower rate than the one presented in \autoref{cor:PurityEvolutionBound}. Those results could easily be adapted to more complex injection functions.

Using the results one the controllability and the purity of the final state, one can choose a particular number of SI layers and particular injection functions according to the desired specifications. In the following Section, we will discuss the impact of the state purity on the final quantum output model. 

    \subsection{Purity and Distinguishability}\label{subsec:PurityDistinguishability}

We define the concept of distinguishability of the quantum model. In what follows, we denote by $\mathrm{Dens}(d)$ the set of $d \times d$ density matrices.

\begin{definition}[Quantum models Distinguishability]
    Given two density matrices $\rho,\sigma \in \mathrm{Dens}(d)$, we introduce the following measure of distinguishability between the two states:
    \begin{equation}\label{eq:def-distinguishability-of-pair}
        \mathcal{D}(\rho, \sigma) := \max_{O} \big|\Tr[O \rho] - \Tr[O \sigma]\big|\,,
    \end{equation}
    where the maximum is taken over all observable $O \in \mathrm{Herm}(d)$ such that $||O||_{\infty} \leq 1$.

    Note that the quantity $\frac{1}{2}\mathcal{D}(\rho,\sigma)$ is equal to the \emph{trace-distance} $D_{\mathrm{tr}}(\rho,\sigma):=\frac{1}{2}\norm{\rho - \sigma}_1$ (see \autoref{chap:ProofTheoremIndiscernability} for details).

    Given now an arbitrary subset $\mathcal{S} \subseteq \mathrm{Dens}(d)$ of density matrices, we define the associated distinguishability measure over $\mathcal{S}$:
    \begin{equation}\label{eq:def-distinguishability-of-subset}
        \mathcal{D}(\mathcal{S}) := \max_{\rho, \sigma \in \mathcal{S}} \mathcal{D}(\rho,\sigma)\,.
    \end{equation}

    Lastly, given a \textit{quantum model} consisting of the output $\rho_{\mathrm{out}}(\bm{\theta})$ of a parametrized quantum circuit, we define the \textit{distinguishability of the quantum model} as:
    \begin{equation}\label{eq:def-distinguishability-of-model}
        \mathcal{D}( \rho_{\mathrm{out}}(\bm{\theta}) ) := \mathcal{D}\Big( \mathcal{S}\!=\!\left\{ 
\rho_{\mathrm{out}}(\bm{\theta}) \; | \; \bm{\theta} \in \Theta \right\} \Big)\,.
    \end{equation}
\end{definition}

A quantum model's distinguishability is an important metric to consider, as a low value would indicate a need for a high number of shots of the whole quantum experiment in order to resolve to sufficient precision the value of the observable that one wishes to estimate.

Intuitively, the connection between the distinguishability of two states $\mathcal{D}(\rho,\sigma)$ and their purity is clear: If two states are both too impure, they must both be relatively close to the maximally-mixed state, and therefore they should be relatively close to each other in some measure of distinguishability. This intuition is quantified in the following result:  


\begin{restatable}{thm}{Distinguishability}
\label{thm:Distinguishability}
    If $\mathcal{S} \subseteq \mathrm{Dens}(d)$ is a subset of $d \times d$ density matrices that are bounded in purity by $\gamma \in [1/d,\ 1]$, i.e., for all $\rho \in \mathcal{S}$,
    \begin{equation}
        \Tr[\rho^2] \leq \gamma,
    \end{equation}
     then the distinguishability $\mathcal{D}(\mathcal{S})$ of this subset (see \autoref{eq:def-distinguishability-of-subset}) satisfies
     \begin{equation}
         \mathcal{D}(\mathcal{S}) \leq 2\sqrt{d}\,\,\sqrt{ \gamma - \frac{1}{d} }\,.
     \end{equation}

\end{restatable}

The proof is given in \autoref{chap:ProofTheoremIndiscernability}.

\section{Probability Estimation Problem}\label{sec:Proba_Estimation}

In this Section, we argue on how linear optics boosted with photonic SI allows one to generate output probabilities that are believed hard to solve classically. Similar task has been tackled using ALO in \cite{chabaud_quantum_2021}. Here, we show how our proposal, which is less experimentally challenging, does offer a similar kind of advantage.

While Boson Sampling is -- under widely believed complexity theoretic assumptions --  a task intractable for a classical computer \cite{aaronson_computational_2011}, estimating a single output probability up to an inverse polynomial additive error is doable in polynomial time (in the number of involved photons) via Gurvits's algorithm \cite{aaronson_generalizing_2012}. This requires QML algorithms relying on Boson Sampling probabilities a bit of work to escape this. It was shown in many ways how a Boson Sampling-like architecture can be made universal for quantum computation \cite{knill_scheme_2001,bartolucci_fusion-based_2023,chabaud_quantum_2021}. We propose an architecture similar to \cite{chabaud_quantum_2021}, where Boson Sampling was boosted with \emph{feedforward} of measurement result. This architecture escapes Gurvits's algorithm for probability estimation, depending on both the number of feedforwards and the number of photon measured throughout the computation (see \autoref{fig:aM}). Incidentally, this scheme becomes universal when both quantities are large (in fact, at least linear) compared to the number of input modes. However, this model comes with strong technical requirements when it comes to physical implementation. Indeed, parameterizing unitaries depending on previous measurement result, \emph{i.e.,} on the fly with respect to the ongoing computation, is out of reach of current technologies, due, in particular, to the heavy requirements for high speed electronics and information processing as explained in \autoref{sec:State_Injection}. This comes at the cost of the universality of the scheme, as we leave the question of the computational universality of the scheme open.

\begin{figure}[h!t]
\centering
\begin{subfigure}{.5\textwidth}
    \centering
    \includegraphics[align=c,width=.9\linewidth]{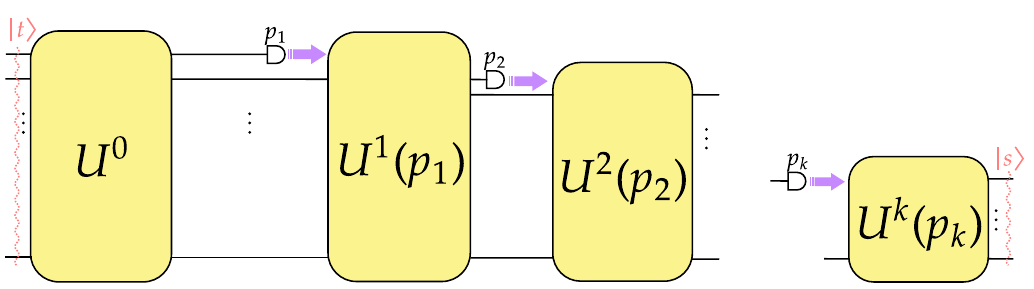}
    \caption{A cascade of $k$ linear optical interferometer, where each of the $U^l$, $1 \leq l \leq k$ is parameterized by the outcome of the measurement of the first mode of the previous interferometer $U^{l-1}$. 
    }
    \label{fig:aM}
\end{subfigure}%
\hspace*{.2in}
\begin{subfigure}{.5\textwidth}
      \centering
    \includegraphics[align=c,width=0.9\linewidth]{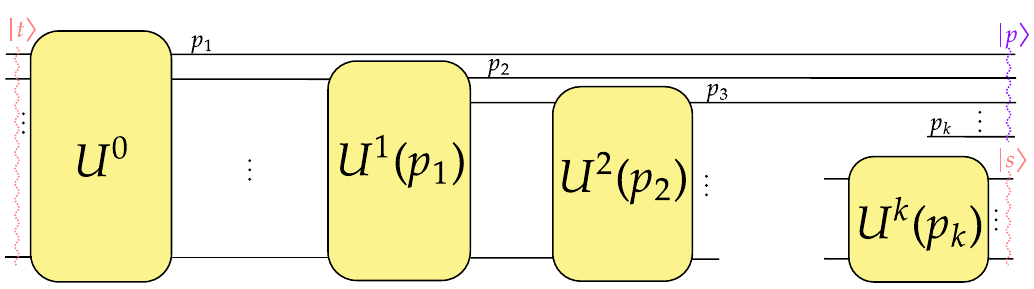}
    \caption{\emph{Equivalent model}, where no feedforward is required. We consider the same unitaries, and we post-select measurement of the modes previously used for adaptivity.}
    \vfill
    \label{fig:aMEM}
\end{subfigure}
\caption{Feedforward architecture (\autoref{fig:aM}) and its equivalent model (\autoref{fig:aMEM}). }
\label{fig:aSI}
\end{figure}

For a computational model, we refer its \emph{equivalent model} (see in \autoref{fig:aSI} as the one describing the transition amplitudes of the original computational model by a (or possibly a sum of) transition amplitude arising from a Boson Sampling experiment. Accordingly, the \emph{equivalent unitary} is the unitary matrix describing that particular experiment state. Indeed, real-life experiments are to be conducted by implementing the actual computational model; the equivalent model is of great help to grasp the hardness of its classical simulation.

\begin{figure}[h!t]
\centering
\begin{subfigure}{.5\textwidth}
    \centering
    \includegraphics[align=c,width=.9\linewidth]{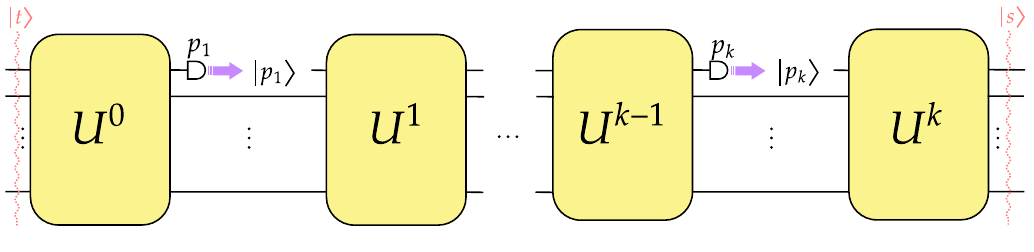}
    \caption{Probability estimation using state injection, with specific injection functions. 
    }
    \label{fig:ProbaStateInjection}
\end{subfigure}%
\hspace*{.2in}
\begin{subfigure}{.5\textwidth}
    \centering
    \includegraphics[align=c,width=.9\linewidth]{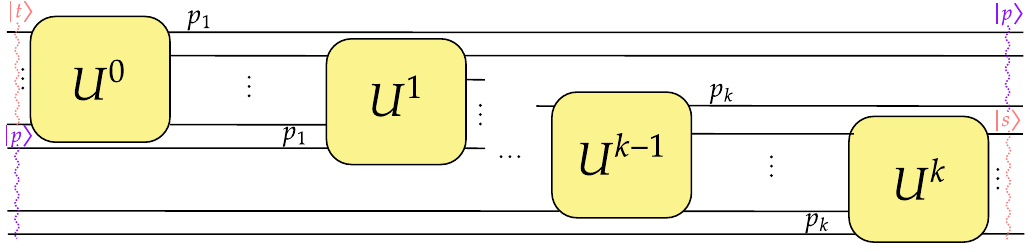}
    \caption{Equivalent model of SI presented in \autoref{fig:Experimental_Comparison_StateInjection}.}
    \vfill
    \label{fig:EqStateInj}
\end{subfigure}
\caption{Feedforward architecture (\autoref{fig:aM}) and its equivalent model (\autoref{fig:aMEM}). }
\label{fig:aSIEM}
\end{figure}

Writing $\Id_l$ the identity on $SU(l)$, the equivalent unitary -- the one described in \autoref{fig:aSIEM} -- of a computation with $k$ state injections is of the form
\begin{equation}
    \tilde U = \prod_{l=0}^k(\Id_l \otimes U^l \otimes \Id_{k-l}),
\end{equation}
where the $U^l \in SU(m)$ for all $0\leq l\leq k$ and $\tilde U \in SU(m+k)$. Thus, the output state reads
\begin{equation}
    \rho_{\mathrm{out}} = \Tr_k \lbr \tilde U \lpr\ket{\t}\bra{\t} \otimes \ket{\p}\bra{\p}\rpr\tilde U^{\dagger} \lpr\Id_{m-k}\otimes \ket{\p}\bra{\p}\rpr \rbr,
\end{equation}
where the $k$ last modes are traced out.
We write $\pr[\t]{\p, \s}$, for $\p \in \Phi_{k,r}$, where $r\in\mathbb N$ is the total number of measured photons and $\bm s,\bm t \in \Phi_{m,n}$ the probability of measuring output occupancy $\s$ upon obtaining the measurement pattern $\p$; indeed we get from the equivalent model that
\begin{equation}
    \pr[\t]{\p, \s} = \bra{(\s, \p)} \tilde U \ket{(\p, \t)} = \frac{1}{(\p!)^2\s!\t!}\lb \per{I_{(\s,\p)}^{\dagger}\tilde U I_{(\p,\t)}}\rb^2.
\end{equation}
Thus, the probability of measuring $\s$ at the output of the interferometer is obtained by summing over all patterns, namely
\begin{equation}
    \pr[\t]{\s} = \sum_{\p \in \Phi_{k,r}} \pr[\t]{\p, \s} = \frac{1}{\s!\t!} \sum_{\p \in \Phi_{k,r}}  \frac{1}{(\p!)^2}\lb \per{I_{(\s,\p)}^{\dagger}\tilde U I_{(\p,\t)}}\rb^2,
\end{equation}
where $\Phi_{k,r}$ is the list of all possible patterns of $k$ measurements where a total of $r$ photons were measured, \emph{i.e.,} the set of all $k$-tuples of integers $(n_1, \cdots, n_k)$ such that $\sum_i n_i = r$. Recall that $|\Phi_{k,r}| = \binom{k+r-1}{r-1}$. Given an $n \times n$ matrix $A$, Gurvits's algorithm computes an estimate of $\per{A}$ to  within additive precision $\pm \varepsilon \|A\|^n$ in time $O(n^2 \varepsilon^{-2})$. Therefore, as $\tilde U$ is unitary, one can compute an estimate of $\pr[\t]{\s}$ in polynomial time up to inverse polynomial precision provided $|\Phi_{k,r}| = O(poly(m))$. 
On the contrary, one will always be able to compute an estimate (again, up to inverse polynomial precision) of the output probability of a certain state by sampling the linear interferometer.

\begin{table}[h!t]
\begin{minipage}[b]{.45\textwidth}
\centering
\begin{tabular}{cc|ccc|}
\cline{3-5}
 &  & \multicolumn{3}{c|}{$k$} \\ \cline{3-5} 
 &  & \multicolumn{1}{c|}{$\ \bigo 1\ $} & \multicolumn{1}{c|}{$\bigo{\log m}$} & $\bigo{m}$ \\ \hline
\multicolumn{1}{|c|}{} & $\bigo{1}$ & \multicolumn{1}{c|}{\cmark} & \multicolumn{1}{c|}{\cmark} & \cmark \\ \cline{2-5} 
\multicolumn{1}{|c|}{} & $\bigo{\log m}$ & \multicolumn{1}{c|}{\cmark} & \multicolumn{1}{c|}{\cmark} & \xmark \\ \cline{2-5} 
\multicolumn{1}{|c|}{} & $\bigo{m}$ & \multicolumn{1}{c|}{\cmark} & \multicolumn{1}{c|}{\xmark} & \xmark \\ \cline{2-5} 
\multicolumn{1}{|c|}{} & $\bigo{m\log m}$ & \multicolumn{1}{c|}{\cellcolor[HTML]{C0C0C0}{\color[HTML]{C0C0C0} }} & \multicolumn{1}{c|}{\xmark} & \xmark \\ \cline{2-5} 
\multicolumn{1}{|c|}{\multirow{-5}{*}{$r$}} & $\bigo{m^2}$ & \multicolumn{1}{c|}{\cellcolor[HTML]{C0C0C0}{\color[HTML]{C0C0C0} }} & \multicolumn{1}{c|}{\cellcolor[HTML]{C0C0C0}{\color[HTML]{C0C0C0} }} & \xmark \\ \hline
\end{tabular}
\vspace{.4em}
\caption{Efficient classical output probability estimation regimes of the SI scheme. The green checkmarks indicate that the output probabilities can be estimated efficiently classically while the red crosses indicate no efficient classical algorithm is known. The gray cells indicate non reachable regimes assuming no concentration of the photons.}
\label{tbl:simulationRegimes}
\end{minipage}
\hspace{.08\textwidth}
\begin{minipage}[b]{.45\textwidth}
\centering
\begin{tabular}{cc|ccc|}
\cline{3-5}
 &  & \multicolumn{3}{c|}{$k$} \\ \cline{3-5} 
 &  & \multicolumn{1}{c|}{$\bigo 1$} & \multicolumn{1}{c|}{$\bigo{\log m}$} & $\bigo{m}$ \\ \hline
\multicolumn{1}{|c|}{\multirow{3}{*}{$r$}} & $\bigo{1}$ & \multicolumn{1}{c|}{\cmark} & \multicolumn{1}{c|}{\cmark} & \cmark \\ \cline{2-5} 
\multicolumn{1}{|c|}{} & $\bigo{\log m}$ & \multicolumn{1}{c|}{\cmark} & \multicolumn{1}{c|}{\cmark} & \xmark \\ \cline{2-5} 
\multicolumn{1}{|c|}{} & $\bigo{m}$ & \multicolumn{1}{c|}{\cmark} & \multicolumn{1}{c|}{\xmark} & \xmark \\ \hline
\end{tabular}
\vspace{.4em}
\caption{Duplicate of \cite[Table 2]{chabaud_quantum_2021}. It describes the simulability regimes for probability estimation in the setting of ALO.  The green checkmarks indicate that the output probabilities can be estimated efficiently classically while the red crosses indicate no efficient classical algorithm is known. \\}
\label{tbl:simulationRegimesALO}
\end{minipage}
\end{table}

The regimes where this technique can be efficiently simulated classically are summarized in \autoref{tbl:simulationRegimes}. We recall in \autoref{tbl:simulationRegimesALO} the regimes in which probability estimation can be efficiently done classically in the ALO picture introduced in \cite{chabaud_quantum_2021}. We remark that our scheme presents the same potential of quantum speedup in the settings where $r$, the total number of photon measured throughout the computation, is smaller than $\bigo m$. But as we offer the possibility to inject photon states anywhere in the computation, the total number of photons to be considered in our scheme goes up to $O(m^2)$. Since we start with the same input state as \cite{chabaud_quantum_2021}, some regimes (grayed in \autoref{tbl:simulationRegimes}) are out of reach. Nonetheless, our proposal gives access to new computing regimes where classical simulation is not expected to be efficient. In addition, using the SI scheme allows us to start with a lower number of initial photons $n_0$. In the case of ALO, $n_0$ must be greater or equal to $r$, which is not the case for SI. As explained in \autoref{sec:State_Injection}, this could help to reach experimentally a regime where probability estimation is not believed to be efficiently classically estimated.
\section{Conclusion}\label{sec:Conclusion}

In this work, we propose a new protocol called State Injection. This new adaptive operation allows us to increase the controllability and thus the expressivity of quantum models in comparison with Linear Optics circuits. We show that this method could be easier to implement than the ALO scheme. While the question of what set of states can reach the quantum models made of linear optics and SI operations is still open, we propose a study of the controllability of those states, and we offered theoretical tools to consider the effect of those new operations on their controllability and purity. The ALO scheme allows us to design universal architecture \cite{knill_scheme_2001}, and future works should investigate if a universal scheme can be designed based on SI.

Finally, we discussed the opportunity of using SI to design subspace-preserving circuits. We gave an example of injection functions that preserve the number of particles. Considering such subspace preserving variational circuits may offer theoretical guarantees on the trainability. Overall, this paper opens a new path to designing useful quantum algorithms using sub-universal circuits.
\section{Acknowledgment}

This work has been co-funded by the H2020-FETOPEN Grant PHOQUSING (GA no.: 899544) and the Naval Group Centre of Excellence for Information Human factors and Signature Management (CEMIS).
H.T. acknowledges co-funding by the European Commission as part of the EIC accelerator program under the grant agreement 190188855 for SEPOQC project, by the Horizon-CL4 program under the grant agreement 101135288 for EPIQUE project, and by the CIFRE grant n°2023/1746.
The authors warmly thank Beatrice Polacchi, Jonas Landman, and Pierre-Emmanuel Emeriau for the fruitful discussions.

\section{Author Contributions}

L.M. initiated the project, proposed the state injection method and focused on the link with the purity of the state and its controllability. E.M. came up with the results on the connection between the final state purity and the quantum model distinguishability. H.T. proposed the idea of probability estimation using state injection, and did the proof with L.M., V.Y., and E.M. V.Y. initiated the project and contributed to the design of the state injection method. The project was supervised by U.C. and E.K..

\printbibliography 

@article{chabaud_quantum_2021,
  title        = {Quantum machine learning with adaptive linear optics},
  volume       = {5},
  url          = {https://doi.org/10.22331%2Fq-2021-07-05-496},
  doi          = {10.22331/q-2021-07-05-496},
  pages        = {496},
  journaltitle = {Quantum},
  author       = {Chabaud, Ulysse and Markham, Damian and Sohbi, Adel},
  date         = {2021-07},
  note         = {Publisher: Verein zur Forderung des Open Access Publizierens in den Quantenwissenschaften}
}

@article{knill_quantum_2002,
  title    = {Quantum gates using linear optics and postselection},
  volume   = {66},
  url      = {https://link.aps.org/doi/10.1103/PhysRevA.66.052306},
  doi      = {10.1103/PhysRevA.66.052306},
  number   = {5},
  urldate  = {2023-11-27},
  journal  = {Physical Review A},
  author   = {Knill, E.},
  month    = nov,
  year     = {2002},
  note     = {Publisher: American Physical Society},
  pages    = {052306}
}

@misc{aaronson_generalizing_2012,
	title = {Generalizing and {Derandomizing} {Gurvits}'s {Approximation} {Algorithm} for the {Permanent}},
	url = {http://arxiv.org/abs/1212.0025},
	urldate = {2023-03-15},
	publisher = {arXiv},
	author = {Aaronson, Scott and Hance, Travis},
	month = dec,
	year = {2012},
	note = {arXiv:1212.0025 [quant-ph]},
}

@article{knill_scheme_2001,
	title = {A scheme for efficient quantum computation with linear optics},
	volume = {409},
	copyright = {2001 Macmillan Magazines Ltd.},
	issn = {1476-4687},
	url = {https://www.nature.com/articles/35051009},
	doi = {10.1038/35051009},
	language = {en},
	number = {6816},
	urldate = {2024-07-17},
	journal = {Nature},
	author = {Knill, E. and Laflamme, R. and Milburn, G. J.},
	month = jan,
	year = {2001},
	note = {Publisher: Nature Publishing Group},
	pages = {46--52},
}

@article{bartolucci_fusion-based_2023,
	title = {Fusion-based quantum computation},
	volume = {14},
	copyright = {2023 The Author(s)},
	issn = {2041-1723},
	url = {https://www.nature.com/articles/s41467-023-36493-1},
	doi = {10.1038/s41467-023-36493-1},
	language = {en},
	number = {1},
	urldate = {2024-02-05},
	journal = {Nature Communications},
	author = {Bartolucci, Sara and Birchall, Patrick and Bombín, Hector and Cable, Hugo and Dawson, Chris and Gimeno-Segovia, Mercedes and Johnston, Eric and Kieling, Konrad and Nickerson, Naomi and Pant, Mihir and Pastawski, Fernando and Rudolph, Terry and Sparrow, Chris},
	month = feb,
	year = {2023},
	note = {Number: 1
Publisher: Nature Publishing Group},
	pages = {912},
}

@inbook{Gard_2015,
   title={An Introduction to Boson-Sampling},
   ISBN={9789814678704},
   url={http://dx.doi.org/10.1142/9789814678704_0008},
   DOI={10.1142/9789814678704_0008},
   booktitle={From Atomic to Mesoscale},
   publisher={WORLD SCIENTIFIC},
   author={Gard, Bryan T. and Motes, Keith R. and Olson, Jonathan P. and Rohde, Peter P. and Dowling, Jonathan P.},
   year={2015},
   month=jun, pages={167–192} }

@article{Bremner2010ClassicalSO,
  title={Classical simulation of commuting quantum computations implies collapse of the polynomial hierarchy},
  author={Michael J. Bremner and Richard Jozsa and Dan J. Shepherd},
  journal={Proceedings of the Royal Society A: Mathematical, Physical and Engineering Sciences},
  year={2010},
  volume={467},
  pages={459 - 472},
  url={https://api.semanticscholar.org/CorpusID:12301677}
}

@article{Zhong_2020,
   title={Quantum computational advantage using photons},
   volume={370},
   ISSN={1095-9203},
   url={http://dx.doi.org/10.1126/science.abe8770},
   DOI={10.1126/science.abe8770},
   number={6523},
   journal={Science},
   publisher={American Association for the Advancement of Science (AAAS)},
   author={Zhong, Han-Sen and Wang, Hui and Deng, Yu-Hao and Chen, Ming-Cheng and Peng, Li-Chao and Luo, Yi-Han and Qin, Jian and Wu, Dian and Ding, Xing and Hu, Yi and Hu, Peng and Yang, Xiao-Yan and Zhang, Wei-Jun and Li, Hao and Li, Yuxuan and Jiang, Xiao and Gan, Lin and Yang, Guangwen and You, Lixing and Wang, Zhen and Li, Li and Liu, Nai-Le and Lu, Chao-Yang and Pan, Jian-Wei},
   year={2020},
   month=dec, pages={1460–1463} }

@article{Preskill_2018,
   title={Quantum Computing in the NISQ era and beyond},
   volume={2},
   ISSN={2521-327X},
   url={http://dx.doi.org/10.22331/q-2018-08-06-79},
   DOI={10.22331/q-2018-08-06-79},
   journal={Quantum},
   publisher={Verein zur Forderung des Open Access Publizierens in den Quantenwissenschaften},
   author={Preskill, John},
   year={2018},
   month=aug, pages={79} }

@misc{polino2022photonicimplementationquantumgravity,
      title={Photonic Implementation of Quantum Gravity Simulator}, 
      author={Emanuele Polino and Beatrice Polacchi and Davide Poderini and Iris Agresti and Gonzalo Carvacho and Fabio Sciarrino and Andrea Di Biagio and Carlo Rovelli and Marios Christodoulou},
      year={2022},
      eprint={2207.01680},
      archivePrefix={arXiv},
      primaryClass={quant-ph},
      url={https://arxiv.org/abs/2207.01680}, 
}

@misc{steinbrecher2018,
      title={Quantum optical neural networks}, 
      author={Gregory R. Steinbrecher and Jonathan P. Olson and Dirk Englund and Jacques Carolan},
      year={2018},
      eprint={1808.10047},
      archivePrefix={arXiv},
      primaryClass={quant-ph},
      url={https://arxiv.org/abs/1808.10047}, 
}

@article{fu2023photonic,
  title={Photonic machine learning with on-chip diffractive optics},
  author={Fu, Tingzhao and Zang, Yubin and Huang, Yuyao and Du, Zhenmin and Huang, Honghao and Hu, Chengyang and Chen, Minghua and Yang, Sigang and Chen, Hongwei},
  journal={Nature Communications},
  volume={14},
  number={1},
  pages={70},
  year={2023},
  publisher={Nature Publishing Group UK London}
}

@article{taballione202320,
  title={20-mode universal quantum photonic processor},
  author={Taballione, Caterina and Anguita, Malaquias Correa and de Goede, Michiel and Venderbosch, Pim and Kassenberg, Ben and Snijders, Henk and Kannan, Narasimhan and Vleeshouwers, Ward L and Smith, Devin and Epping, J{\"o}rn P and others},
  journal={Quantum},
  volume={7},
  pages={1071},
  year={2023},
  publisher={Verein zur F{\"o}rderung des Open Access Publizierens in den Quantenwissenschaften}
}

@article{thomas2024noise,
  title={Noise Performance of On-Chip Nano-Mechanical Switches for Quantum Photonics Applications},
  author={Thomas, Rodrigo A and Qvotrup, Celeste and Liu, Zhe and Midolo, Leonardo},
  journal={Advanced Quantum Technologies},
  pages={2400012},
  year={2024},
  publisher={Wiley Online Library}
}

@article{memeo2024micro,
  title={Micro-opto-mechanical glass interferometer for megahertz modulation of optical signals},
  author={Memeo, Roberto and Crespi, Andrea and Osellame, Roberto},
  journal={Optica},
  volume={11},
  number={2},
  pages={178--183},
  year={2024},
  publisher={Optica Publishing Group}
}

@article{calvarese2022strategies,
  title={Strategies for improved temporal response of glass-based optical switches},
  author={Calvarese, Matteo and Pai{\`e}, Petra and Ceccarelli, Francesco and Sala, Federico and Bassi, Andrea and Osellame, Roberto and Bragheri, Francesca},
  journal={Scientific Reports},
  volume={12},
  number={1},
  pages={239},
  year={2022},
  publisher={Nature Publishing Group UK London}
}

@article{tian2024piezoelectric,
  title={Piezoelectric actuation for integrated photonics},
  author={Tian, Hao and Liu, Junqiu and Kippenberg, Tobias J and Bhave, Sunil},
  journal={arXiv preprint arXiv:2405.08836},
  year={2024}
}

@article{maring2024versatile,
  title={A versatile single-photon-based quantum computing platform},
  author={Maring, Nicolas and Fyrillas, Andreas and Pont, Mathias and Ivanov, Edouard and Stepanov, Petr and Margaria, Nico and Hease, William and Pishchagin, Anton and Lema{\^\i}tre, Aristide and Sagnes, Isabelle and others},
  journal={Nature Photonics},
  volume={18},
  number={6},
  pages={603--609},
  year={2024},
  publisher={Nature Publishing Group UK London}
}

@article{wang2017high,
  title={High-efficiency multiphoton boson sampling},
  author={Wang, Hui and He, Yu and Li, Yu-Huai and Su, Zu-En and Li, Bo and Huang, He-Liang and Ding, Xing and Chen, Ming-Cheng and Liu, Chang and Qin, Jian and others},
  journal={Nature Photonics},
  volume={11},
  number={6},
  pages={361--365},
  year={2017},
  publisher={Nature Publishing Group UK London}
}

@article{zhong201812,
  title={12-photon entanglement and scalable scattershot boson sampling with optimal entangled-photon pairs from parametric down-conversion},
  author={Zhong, Han-Sen and Li, Yuan and Li, Wei and Peng, Li-Chao and Su, Zu-En and Hu, Yi and He, Yu-Ming and Ding, Xing and Zhang, Weijun and Li, Hao and others},
  journal={Physical review letters},
  volume={121},
  number={25},
  pages={250505},
  year={2018},
  publisher={APS}
}

@article{sulimany2024quantum,
  title={Quantum-secure multiparty deep learning},
  author={Sulimany, Kfir and Vadlamani, Sri Krishna and Hamerly, Ryan and Iyengar, Prahlad and Englund, Dirk},
  journal={arXiv preprint arXiv:2408.05629},
  year={2024}
}

@article{hong1987measurement,
  title={Measurement of subpicosecond time intervals between two photons by interference},
  author={Hong, Chong-Ki and Ou, Zhe-Yu and Mandel, Leonard},
  journal={Physical review letters},
  volume={59},
  number={18},
  pages={2044},
  year={1987},
  publisher={APS}
}

@article{PRXQuantum.3.010313,
  title = {Connecting Ansatz Expressibility to Gradient Magnitudes and Barren Plateaus},
  author = {Holmes, Zo\"e and Sharma, Kunal and Cerezo, M. and Coles, Patrick J.},
  journal = {PRX Quantum},
  volume = {3},
  issue = {1},
  pages = {010313},
  numpages = {23},
  year = {2022},
  month = jan,
  publisher = {American Physical Society},
  doi = {10.1103/PRXQuantum.3.010313},
  url = {https://link.aps.org/doi/10.1103/PRXQuantum.3.010313}
}

@misc{xiong2023,
      title={On fundamental aspects of quantum extreme learning machines}, 
      author={Weijie Xiong and Giorgio Facelli and Mehrad Sahebi and Owen Agnel and Thiparat Chotibut and Supanut Thanasilp and Zoë Holmes},
      year={2023},
      eprint={2312.15124},
      archivePrefix={arXiv},
      primaryClass={quant-ph},
      url={https://arxiv.org/abs/2312.15124}, 
}

@misc{mhiri2024,
      title={Constrained and Vanishing Expressivity of Quantum Fourier Models}, 
      author={Hela Mhiri and Leo Monbroussou and Mario Herrero-Gonzalez and Slimane Thabet and Elham Kashefi and Jonas Landman},
      year={2024},
      eprint={2403.09417},
      archivePrefix={arXiv},
      primaryClass={quant-ph},
      url={https://arxiv.org/abs/2403.09417}, 
}

@article{Larocca_2023,
   title={Theory of overparametrization in quantum neural networks},
   volume={3},
   ISSN={2662-8457},
   url={http://dx.doi.org/10.1038/s43588-023-00467-6},
   DOI={10.1038/s43588-023-00467-6},
   number={6},
   journal={Nature Computational Science},
   publisher={Springer Science and Business Media LLC},
   author={Larocca, Martín and Ju, Nathan and García-Martín, Diego and Coles, Patrick J. and Cerezo, Marco},
   year={2023},
   month=jun, pages={542–551} }

@article{Larocca_2022,
   title={Diagnosing Barren Plateaus with Tools from Quantum Optimal Control},
   volume={6},
   ISSN={2521-327X},
   url={http://dx.doi.org/10.22331/q-2022-09-29-824},
   DOI={10.22331/q-2022-09-29-824},
   journal={Quantum},
   publisher={Verein zur Forderung des Open Access Publizierens in den Quantenwissenschaften},
   author={Larocca, Martin and Czarnik, Piotr and Sharma, Kunal and Muraleedharan, Gopikrishnan and Coles, Patrick J. and Cerezo, M.},
   year={2022},
   month=sep, pages={824} 
}

@misc{Ragone2023,
      title={A Unified Theory of Barren Plateaus for Deep Parametrized Quantum Circuits}, 
      author={Michael Ragone and Bojko N. Bakalov and Frédéric Sauvage and Alexander F. Kemper and Carlos Ortiz Marrero and Martin Larocca and M. Cerezo},
      year={2023},
      eprint={2309.09342},
      archivePrefix={arXiv},
      primaryClass={quant-ph}
}

@misc{Fontana2023,
      title={The Adjoint Is All You Need: Characterizing Barren Plateaus in Quantum Ans\"atze}, 
      author={Enrico Fontana and Dylan Herman and Shouvanik Chakrabarti and Niraj Kumar and Romina Yalovetzky and Jamie Heredge and Shree Hari Sureshbabu and Marco Pistoia},
      year={2023},
      eprint={2309.07902},
      archivePrefix={arXiv},
      primaryClass={quant-ph}
}

@misc{monbroussou2023,
      title={Trainability and Expressivity of Hamming-Weight Preserving Quantum Circuits for Machine Learning}, 
      author={Léo Monbroussou and Jonas Landman and Alex B. Grilo and Romain Kukla and Elham Kashefi},
      year={2023},
      eprint={2309.15547},
      archivePrefix={arXiv},
      primaryClass={quant-ph},
      url={https://arxiv.org/abs/2309.15547}, 
}

@article{McClean2018,
	doi = {10.1038/s41467-018-07090-4},
  
	url = {https://doi.org/10.1038%2Fs41467-018-07090-4},
  
	year = 2018,
	month = nov,
  
	publisher = {Springer Science and Business Media {LLC}
},
	volume = {9},
  
	number = {1},
  
	author = {Jarrod R. McClean and Sergio Boixo and Vadim N. Smelyanskiy and Ryan Babbush and Hartmut Neven},
  
	title = {Barren plateaus in quantum neural network training landscapes},
  
	journal = {Nature Communications}
}

@misc{diaz2023,
      title={Showcasing a Barren Plateau Theory Beyond the Dynamical Lie Algebra}, 
      author={N. L. Diaz and Diego García-Martín and Sujay Kazi and Martin Larocca and M. Cerezo},
      year={2023},
      eprint={2310.11505},
      archivePrefix={arXiv},
      primaryClass={quant-ph},
      url={https://arxiv.org/abs/2310.11505}, 
}

@misc{cerezo2024does,
      title={Does provable absence of barren plateaus imply classical simulability? Or, why we need to rethink variational quantum computing}, 
      author={M. Cerezo and Martin Larocca and Diego García-Martín and N. L. Diaz and Paolo Braccia and Enrico Fontana and Manuel S. Rudolph and Pablo Bermejo and Aroosa Ijaz and Supanut Thanasilp and Eric R. Anschuetz and Zoë Holmes},
      year={2024},
      eprint={2312.09121},
      archivePrefix={arXiv},
      primaryClass={quant-ph}
}

@misc{goh2023,
      title={Lie-algebraic classical simulations for variational quantum computing}, 
      author={Matthew L. Goh and Martin Larocca and Lukasz Cincio and M. Cerezo and Frédéric Sauvage},
      year={2023},
      eprint={2308.01432},
      archivePrefix={arXiv},
      primaryClass={quant-ph},
      url={https://arxiv.org/abs/2308.01432}, 
}

@misc{coyle2024trainingefficientdensityquantummachine,
      title={Training-efficient density quantum machine learning}, 
      author={Brian Coyle and El Amine Cherrat and Nishant Jain and Natansh Mathur and Snehal Raj and Skander Kazdaghli and Iordanis Kerenidis},
      year={2024},
      eprint={2405.20237},
      archivePrefix={arXiv},
      primaryClass={quant-ph},
      url={https://arxiv.org/abs/2405.20237}, 
}

@article{Cong_2019,
   title={Quantum convolutional neural networks},
   volume={15},
   ISSN={1745-2481},
   url={http://dx.doi.org/10.1038/s41567-019-0648-8},
   DOI={10.1038/s41567-019-0648-8},
   number={12},
   journal={Nature Physics},
   publisher={Springer Science and Business Media LLC},
   author={Cong, Iris and Choi, Soonwon and Lukin, Mikhail D.},
   year={2019},
   month=aug, pages={1273–1278} }

@article{Bamber-HowMany-1985,
	title = {How many parameters can a model have and still be testable?},
	volume = {29},
	url = {https://www.sciencedirect.com/science/article/pii/0022249685900057},
	doi = {https://doi.org/10.1016/0022-2496(85)90005-7},
	pages = {443--473},
	number = {4},
	journaltitle = {Journal of Mathematical Psychology},
	author = {Bamber, Donald and Van Santen, Jan {PH}},
	date = {1985},
}

@article{Reck-ExperimentalRealization-1994,
  title = {Experimental realization of any discrete unitary operator},
  author = {Reck, Michael and Zeilinger, Anton and Bernstein, Herbert J. and Bertani, Philip},
  journal = {Phys. Rev. Lett.},
  volume = {73},
  issue = {1},
  pages = {58--61},
  numpages = {0},
  year = {1994},
  month = Jul,
  publisher = {American Physical Society},
  doi = {10.1103/PhysRevLett.73.58},
  url = {https://link.aps.org/doi/10.1103/PhysRevLett.73.58}
}

@article{Parellada-NogoTheorems-2023,
	title = {No-go theorems for photon state transformations in quantum linear optics},
	volume = {54},
	issn = {2211-3797},
	url = {https://www.sciencedirect.com/science/article/pii/S2211379723009014},
	doi = {10.1016/j.rinp.2023.107108},
	pages = {107108},
	journaltitle = {Results in Physics},
	shortjournal = {Results in Physics},
	author = {Parellada, Pablo V. and Gimeno i Garcia, Vicent and Moyano-Fernández, Julio José and Garcia-Escartin, Juan Carlos},
	urldate = {2023-12-01},
	date = {2023-11-01},
}

@inproceedings{aaronson_computational_2011,
	address = {San Jose California USA},
	title = {The computational complexity of linear optics},
	isbn = {978-1-4503-0691-1},
	url = {https://dl.acm.org/doi/10.1145/1993636.1993682},
	doi = {10.1145/1993636.1993682},
	language = {en},
	urldate = {2024-07-18},
	booktitle = {Proceedings of the forty-third annual {ACM} symposium on {Theory} of computing},
	publisher = {ACM},
	author = {Aaronson, Scott and Arkhipov, Alex},
	month = jun,
	year = {2011},
	pages = {333--342},
}

@book{Lee-IntroductionSmooth-2012,
	location = {New York, {NY}},
	title = {Introduction to Smooth Manifolds},
	volume = {218},
	isbn = {978-1-4419-9981-8},
	url = {https://link.springer.com/10.1007/978-1-4419-9982-5},
	series = {Graduate Texts in Mathematics},
	publisher = {Springer},
	author = {Lee, John M.},
	urldate = {2023-07-04},
	date = {2012},
	langid = {english},
	doi = {10.1007/978-1-4419-9982-5},
}

@incollection{PyTorch,
title = {PyTorch: An Imperative Style, High-Performance Deep Learning Library},
author = {Paszke, Adam and Gross, Sam and Massa, Francisco and Lerer, Adam and Bradbury, James and Chanan, Gregory and Killeen, Trevor and Lin, Zeming and Gimelshein, Natalia and Antiga, Luca and Desmaison, Alban and Kopf, Andreas and Yang, Edward and DeVito, Zachary and Raison, Martin and Tejani, Alykhan and Chilamkurthy, Sasank and Steiner, Benoit and Fang, Lu and Bai, Junjie and Chintala, Soumith},
booktitle = {Advances in Neural Information Processing Systems 32},
pages = {8024--8035},
year = {2019},
publisher = {Curran Associates, Inc.},
url = {http://papers.neurips.cc/paper/9015-pytorch-an-imperative-style-high-performance-deep-learning-library.pdf}
}

@article{marcus_permanents_1965,
  title   = {Permanents},
  volume  = {72},
  issn    = {00029890},
  url     = {https://www.jstor.org/stable/2313846?origin=crossref},
  doi     = {10.2307/2313846},
  number  = {6},
  urldate = {2023-08-16},
  journal = {The American Mathematical Monthly},
  author  = {Marcus, Marvin and Minc, Henryk},
  month   = jun,
  year    = {1965},
  pages   = {577}
}

\newpage

\begin{appendix}
    \section{Details about the dimension}\label{chap:Details_Dimension}

In the main text we refer to the \textit{dimension} of various subsets of unitary matrices, or of density matrices. Let us precise here what we mean by the word dimension.

To begin, consider the space $SU(d)$ of $d \times d$ unitary matrices. It is a \textit{manifold of dimension} $d^2 -1$. More specifically, when viewing the space of arbitrary complex $d \times d$ matrices as $\mathbb{R}^{2d^2}$ (through concatenation of its columns and splitting each complex scalar into real and imaginary parts), the subset $SU(d)$ is a \textit{submanifold of dimension $d^2-1$} in this ambient Euclidean space of dimension $2d^2$. This can be understood as meaning that the subset $SU(d)$ is a "smooth hypersurface of dimension $d^2 -1$ inside $\mathbb{R}^{2d^2}$", and that in particular, around each unitary "point" $U \in SU(d)$ on this hypersurface there are $\dim(SU(d))=d^2 - 1$ independent directions in which one can locally move while still remaining on the hypersurface. (For the precise meaning, we refer the reader to the definition of \textit{embedded submanifolds}, e.g. in \cite{Lee-IntroductionSmooth-2012}.) Now, given a subset of unitary matrices $\mathcal{W} \subseteq SU(d)$, there are two situations we will consider that will let us talk about the notion of $\dim(\mathcal{W})$.

First, if $\mathcal{W}$ is a (closed) subgroup of $SU(d)$, then it is known \cite[Theorem 20.12]{Lee-IntroductionSmooth-2012} that $\mathcal{W}$ is a legitimate \textit{submanifold} inside $SU(d)$, and hence the number $\dim(\mathcal{W})$ is meaningful.

Second, if $\mathcal{W}$ is the result of an \textit{analytic} parametrization, \ie $\mathcal{W} := \{ U(\bm{\theta}) \,|\, \bm{\theta} \in \Theta\subseteq\mathbb{R}^p \}$ with a parametrization map $U:\Theta\to SU(d)$ that is smooth and analytic, then even though $\mathcal{W}$ will technically not be a submanifold in general, it still holds that "\textit{locally} around \textit{almost every} point of $\mathcal{W}$, it is a submanifold, of dimension $d_U$". Here, the notions of "\textit{locally}" and "\textit{almost every}" correspond respectively to the topology and Lebesgue measure of the parameter space $\Theta$ (not of the ambient space $SU(d)$),
and the number $d_U$ is given by the maximum over $\bm{\theta} \in \Theta$ of the rank of the Jacobian matrix of the parametrization map $U:\Theta\to SU(d)$ calculated at point $\bm{\theta}$. The reason behind this claim is the combination of the fact that images of smooth maps of constant Jacobian rank $r$ are "locally" submanifolds of dimension $r$ \cite[Sections 3 and 4]{Lee-IntroductionSmooth-2012}, and that the Jacobian rank of an analytic map is constant "almost-everywhere" \cite[Prop. B.4]{Bamber-HowMany-1985}.
In that situation, we let $\dim(\mathcal{W}):=d_U$, and it should be thought as quantifying the "local dimension" of the space $\mathcal{W}$. Note that (simply by the fact that a matrix's rank is no greater than its height or its width), we have $\dim(\mathcal{W}) \leq \min(\dim(\Theta),\dim(SU(d))$, which encapsulates two intuitive but important facts, which we state again in words:
\begin{enumerate}
    \item "The dimension can only decrease or remain constant after going through a parametrization --- and it remains constant iff the parametrization is injective";
    \item "The dimension of a subset is always smaller or equal to the dimension of a manifold in which it is included".
\end{enumerate}

If $\mathcal{W}$ is a subset of $d \times d$ density matrices (instead of unitary matrices), the whole discussion of the above paragraph holds identically, as in fact it holds for any smooth manifold $M$ is place of $SU(d)$.

Now, consider the important following situation of a parametrized unitary, \ie a map $U:\Theta\to SU(d)$, which then acts onto a fixed initial density matrix $\rho_{\mathrm{in}}$ to produce a parametrized density matrix: $\rho_{\mathrm{out}}(\bm{\theta}):=U(\bm{\theta}) \rho_{\mathrm{in}} U(\bm{\theta})^\dagger$. Denote by $\mathrm{Dens}(d)$ the set of $d \times d$ density matrices. Since the resulting map $\rho_{\mathrm{out}}:\Theta\to \mathrm{Dens}(d)$ is the composition of the parametrized unitary map with the initial-state evaluation map $V\mapsto V \rho_{\mathrm{in}} V^\dagger$, the parametrized state can be seen as having as its parameter space the image of the first map, $U(\Theta)$. And therefore, from the two numbered facts above, another intuitive but important fact follows here: the dimension of the space of output density matrices is no greater than the dimension of the parametrized unitaries, \ie 
$$\dim( \{\rho_\mathrm{out}(\bm{\theta}) \,|\, \bm{\theta} \in \Theta \} ) \leq \dim( \{ U(\bm{\theta}) \,|\, \bm{\theta} \in \Theta \} )\,.$$

    \section{State Injection and Purity Evolution}\label{chap:Proof_Purity_Evolution}

In this Section, we present the proof of the results on the density operator purity evolution from \autoref{subsec:Purity}.

\subsection{Proof of \autoref{thm:PurityEvolution}}\label{subchap:ProofTheoremPurityEvolution}

We recall the expression of \autoref{thm:PurityEvolution}:

\PurityEvolution*

\begin{proof}
    We consider an initial pure state $\ketbra{\psi_0}{\psi_0}$, which is Fock state of $n$ photons over $m$ modes. We call $\rho$ the state after the state injection and $\gamma(\rho)$ its purity. Then,
    \begin{equation}
        \gamma(\rho) = \Tr[\rho^2] = \Tr[(\sum_{i \in [n]} \Pr[i] \ketbra{\psi^i}{\psi^i})^2],
    \end{equation}
    with $\ket{\psi^i}$ the pure state post-injection corresponding to the measurement outcome $p=i$. 
    \begin{equation}
        \begin{split}
            \gamma(\rho) &= \Tr[\sum_{i \in [n]} \Pr[i]^2 (\ketbra{\psi^i}{\psi^i})^2 + \sum_{i,j \in [n] \mid i \neq j} \Pr[i] \Pr[j] \ketbra{\psi^i}{\psi^i} \ketbra{\psi^j}{\psi^j}] \\
            &= \sum_{i \in [n]}\Pr[i]^2\Tr[(\ketbra{\psi^i}{\psi^i})^2] + \sum_{i,j \in [n] \mid i \neq j} \Pr[i] \Pr[j] \Tr[\ketbra{\psi^i}{\psi^i} \ketbra{\psi^j}{\psi^j}] \\
            &= \sum_{i \in [n]}\Pr[i]^2 + \sum_{i,j \in [n] \mid i \neq j} \Pr[i] \Pr[j] \Tr[\ketbra{\psi^i}{\psi^i} \ketbra{\psi^j}{\psi^j}],
        \end{split}
    \end{equation}
    as for any $i$, $\ketbra{\psi^i}{\psi^i}$ is a pure state. Considering the specific injection function, we have that for all $i \in [n]$
    \begin{equation}
        \ketbra{\psi^i}{\psi^i} = \Pi_i \ketbra{\psi_0}{\psi_0} \Pi_i^{\dagger},
    \end{equation}
    with $\Pi_i$ the projector of the state over all the Fock basis where there are $i$ photons on the measured modes. By definition, for all $i,j\neq i \in [n], \Pi_i^\dagger \Pi_j = 0$.
    Therefore, if follows that
    \begin{equation}
        \gamma(\rho) = \sum_{i \in [n]}\Pr[i]^2.
    \end{equation}
\end{proof}

\subsection{Proof of \autoref{cor:PurityEvolutionBound}}\label{subchap:ProofCorWorstScenario}

We recall the expression of \autoref{cor:PurityEvolutionBound}:

\PurityEvolutionBound*

\begin{proof}
    We consider a circuit with $L$ state injection layers separated by linear optical circuits. We call $\ketbra{\psi_l}{\psi_l}$ the state after the $l$ state injection layer with $l \in [L]$. Notice that we have: 
    \begin{equation}\label{eq:PurityInductiveRelation}
        \ketbra{\psi_l}{\psi_l} = \sum_{i=o}^n U(\bm{\theta}_{l-1}) \ketbra{\psi^i_{l-1}}{\psi^i_{l-1}} U^\dagger( \bm{\theta}_{l-1}) =  \sum_{i=o}^n U(\bm{\theta}_{l-1}) \Pi_i \ketbra{\psi_{l-1}}{\psi_{l-1}} \Pi_i^\dagger U^\dagger( \bm{\theta}_{l-1}) \, \textrm{,}
    \end{equation}
    with $U(\bm{\theta}_{l-1})$ the unitary corresponding to the linear optical circuit that separated the state injection layers and $\Pi_i$ the projector of the state onto all the Fock basis where there are $i$ photons on the measured modes for the considered state injection. By definition, for all $i,j\neq i \in [n], \Pi_i^\dagger \Pi_j = 0$.
    We can thus state that
    \begin{equation}
        \gamma(\rho) = \Tr[\rho^2] = \Tr[(\sum^n_{i=0} \Pr[L,i] \ketbra{\psi^i_L}{\psi^i_L})^2],
    \end{equation}
    with $\Pr[l,i]$ the probability of measuring $i$ photons in the measurement of the $l^{\textrm{th}}$ state injection layer. In a similar way as in the proof of \autoref{thm:PurityEvolution} (see \autoref{subchap:ProofTheoremPurityEvolution}), we can show that
    \begin{equation}
        \gamma(\rho) = \sum^n_{i=0} \Pr[L,i]^2 \Tr[(\ketbra{\psi^i_L}{\psi^i_L})^2].
    \end{equation}
    In this case, $\ketbra{\psi^i_L}{\psi^i_L}$ is not a pure state. Using the inductive relation from \autoref{eq:PurityInductiveRelation}, we conclude that
    \begin{equation}\label{eq:gammaRhoFromSum}
        \gamma(\rho) = \prod_{l=1}^L (\sum^n_{i=0} \Pr[l,i]^2).
    \end{equation}
    Notice that $\sum_{i=0}^n \Pr[l,i] = 1$ for any $l \in [L]$. Therefore, the sum of the square of the probability is lower bounded by
    \begin{equation}\label{eq:sumPrSquaredLB}
        \sum_{i=0}^n \Pr[l,i]^2 \geq \sum_{i=0}^n (\frac{1}{n+1})^2 = \frac{1}{n+1}.
    \end{equation}
    Finally, plugging \autoref{eq:sumPrSquaredLB} into \autoref{eq:gammaRhoFromSum}, we have
    \begin{equation}
        \gamma(\rho) \leq \prod_{l=1}^L \frac{1}{n+1} = \frac{1}{(n+1)^L} \, \textrm{.}
    \end{equation}
\end{proof}

\subsection{Proof of \autoref{cor:PurityNoCollision}}\label{subchap:ProofCorNoCollision}

We recall the expression of \autoref{cor:PurityNoCollision}:

\PurityNoCollision*

\begin{proof}
    We recall the expression of the Boson Birthday Bound introduced in \cite{aaronson_computational_2011}:

    \begin{restatable}{thm}{BosonBirthdayBound}[{Boson Birthday Bound, adapted from \cite[Theorem 72]{aaronson_computational_2011}}]\label{thm:BosonBirthdayBound}
        Recalling that $\mathcal{H}_{m,m}$ is the Haar measure over $m \times m$ unitary matrices,
        \begin{equation}
            \mathbb{E}_{U \in \mathcal{H}_{m,m}}[\Pr[S \in B_{m,n}]] < \frac{2n^2}{m} \, \textrm{.}
        \end{equation}
        With $S$ a Fock state, and $B_{m,n}$ the set of $m$-modes and $n$-particle states where more than one photon can be per mode.  
    \end{restatable}

    We consider a circuit with $L$ state injection layers separated by linear optical circuits. We consider the case where each unitary matrix corresponding to the Linear Optic layers are Haar random matrices. We call $\ketbra{\psi_l}{\psi_l}$ the state after the $l$ state injection layer with $l \in [L]$. As in the proof of \autoref{cor:PurityEvolutionBound}, we consider the same inductive relation for the purity of the states within the circuit given by \autoref{eq:PurityInductiveRelation}. As in this previous proof, we can use \autoref{eq:gammaRhoFromSum} that we recall here:
    \begin{equation*}
        \gamma(\rho) = \prod_{l=1}^L (\sum^n_{i=0} \Pr[l,i]^2).
    \end{equation*}

    Using \autoref{thm:BosonBirthdayBound}, have that for any $l \in [L]$:
    \begin{equation}
        \mathbb{E}_{U \in \mathcal{H}_{m,m}}[\Pr[l,0] + \Pr[l,1]] \geq \frac{m - 2n^2}{m} \, \textrm{.}
    \end{equation}
    Therefore, we have for any $l \in [L]$:
    \begin{equation}
        \mathbb{E}_{U \in \mathcal{H}_{m,m}}[\sum_{i=0}^n \Pr[l,i]^2] \geq \mathbb{E}_{U \in \mathcal{H}_{m,m}}[\sum_{i=0}^1 \Pr[l,i]^2] \geq 2 (\frac{m - 2n^2}{2m})^2 \, \textrm{.} 
    \end{equation}

    We can conclude:
    \begin{equation}
        \mathbb{E}_{U \in \mathcal{H}_{m,m}}[\gamma(\rho) = \prod_{l=1}^L (\sum^n_{i=0} \Pr[l,i]^2]) \geq  (\frac{m - 2n^2}{\sqrt{2} m})^{2 L} \, \textrm{.}
    \end{equation}

\end{proof}

    \section{Proof of \autoref{thm:Distinguishability}}\label{chap:ProofTheoremIndiscernability}

We recall the expression of \autoref{thm:Distinguishability}:

\Distinguishability*

\begin{proof}
Let us begin with some notation.
We denote by $\mathrm{Herm}(d)$ the space of $d \times d$ Hermitian matrices, and given two matrices $A,B \in \mathrm{Herm}(d)$, we introduce the notation $\langle A, B \rangle := \Tr[A B]$ for the Hilbert-Schmidt inner-product in this real vector space. Denote also by $\mathrm{Dens}(d)$ the subset of density matrices.
We recall the definition of the Schatten $p$-norms (as they will occur in the proof, for $p=$ $1$, $2$ and $\infty$): given a matrix ${A \in \mathbb{C}^{d \times d}}$, ${\norm{A}_{p}:=\norm{\bm{\sigma}}_{p}:=(\sigma_1^p + \cdots + \sigma_d^p)^{1/p}}$, where $\bm{\sigma} \in \mathbb{R}^{d}$ denotes the vector of the singular values of $A$. In the case where $A \in \mathrm{Herm}(d)$, denoting by $\bm{\lambda} \in \mathbb{R}^{d}$ the vector of eigenvalues of $A$, it holds that
$\norm{A}_{1} = |\lambda_1| + \cdots + |\lambda_d|$, 
$\norm{A}_{2} = (\lambda_1^2 + \cdots + \lambda_d^2)^{1/2} = \sqrt{\langle A, A \rangle}$, and
$\norm{A}_{\infty} = \max_{i=1,\dots,d} |\lambda_i|$.

Note that given $A \in \mathrm{Herm}(d)$, its purity $\Tr[A^2]$ is by definition just $\norm{A}_2^2$.
Given $\rho,\sigma \in \mathrm{Dens}(d)$ such that
\begin{equation}\label{eq:ProofTheoremIndiscernability-assumption}
\norm{\rho}_2^2 \leq \gamma\text{\ and\ }\norm{\sigma}_2^2 \leq \gamma\,, 
\end{equation}
we have:

\begin{align}
&\mathcal{D}(\rho,\sigma)\nonumber\\[5pt]
&= \max_{\norm{O}_{\infty}=1} \big| \langle \rho - \sigma, O \rangle\big| && \left(\text{\autoref{eq:def-distinguishability-of-pair}}\right)\label{eq:ProofTheoremIndiscernability-eq-distinguish}\\[5pt]
&= \norm{\rho - \sigma}_1 && \left(\text{trace distance identity, explained below}\right) \label{eq:ProofTheoremIndiscernability-eqtoexplain1}\\[5pt]
&\leq \sqrt{d}\, \norm{\rho - \sigma}_2 && \left(\, \norm{\cdot}_1 \leq \sqrt{d}\norm{\cdot}_2 \,\right)\\[5pt]
&\leq \sqrt{d}\, \big( \norm{\rho - \mmstate{d}}_2 + \norm{\sigma - \mmstate{d}}_2 \big) && \left(\text{triangle inequality for $\norm{\cdot}_2$}\right)\\[5pt]
&= \sqrt{d}\,\left( \sqrt{ \norm{\rho - \bm{0}}_2^2 - \norm{\mmstate{d} - \bm{0}}_2^2 }  +  \sqrt{ \norm{\sigma - \bm{0}}_2^2 - \norm{\mmstate{d} - \bm{0}}_2^2 }  \right) && \left(\text{Pythagorean theorem, explained below}\right) \label{eq:ProofTheoremIndiscernability-eqtoexplain2}\\[5pt]
&= \sqrt{d}\,\left( \sqrt{ \norm{\rho}_2^2 - 1/d}  +   \sqrt{ \norm{\sigma}_2^2 - 1/d}  \right) && \left(\, \norm{\mmstate{d}}_2^2 = 1/d \,\right)\\[5pt]
&\leq \sqrt{d}\,\left( \sqrt{ \gamma - 1/d}  +   \sqrt{ \gamma - 1/d}  \right) && \left(\text{By assumption \autoref{eq:ProofTheoremIndiscernability-assumption}}\right)\\[5pt]
&= 2\sqrt{d}\,\,\sqrt{ \gamma - 1/d}\,,
\end{align}
which gives the desired result. We end by giving more explanations for \autoref{eq:ProofTheoremIndiscernability-eqtoexplain1,eq:ProofTheoremIndiscernability-eqtoexplain2}.

\autoref{eq:ProofTheoremIndiscernability-eqtoexplain1} is a standard inequality --- that gives an operational meaning of the \textit{trace distance} of two states $\rho,\sigma$ ($\frac{1}{2}\norm{\rho - \sigma}_1$) in terms of their distinguishability ($\mathcal{D}(\rho,\sigma)$). It holds because of the following two observations. First, by the $(1,\infty)$-Hölder inequality, one has that for all $O\in\mathrm{Herm}(d)$:
\begin{align}
    \big| \langle \rho - \sigma, O \rangle\big| &\leq \norm{\rho - \sigma}_1 \norm{O}_\infty\,,
    \intertext{implying (\autoref{eq:ProofTheoremIndiscernability-eq-distinguish}) that}
    \mathcal{D}(\rho,\sigma) &\leq \norm{\rho - \sigma}_1\,.
\end{align}
Second, by choosing the specific observable $O\in\mathrm{Herm}(d)$ given by $O:=\sum_{i=1}^d \mathrm{sign}(\lambda_i) \ketbra{e_i}$, with ${\rho - \sigma = \sum_{i=1}^d \lambda_i \ketbra{e_i}}$ the eigendecomposition of the Hermitian matrix $\rho - \sigma$, one has $\norm{O}_\infty = 1$ and
\begin{align}
    \langle \rho - \sigma, O \rangle &= \norm{\rho - \sigma}_1\,,
    \intertext{implying (\autoref{eq:ProofTheoremIndiscernability-eq-distinguish}) that}
    \mathcal{D}(\rho,\sigma) &\geq \norm{\rho - \sigma}_1\,.
\end{align}

In \autoref{eq:ProofTheoremIndiscernability-eqtoexplain2}, we apply the Pythagorean theorem (in the real vector space $\mathrm{Herm}(d)$ with the geometry given by the Hilbert-Schmidt inner-product) to two right triangles. These are respectively the triangles made of the vertices $(\bm{0}, \mmstate{d}, \rho)$ and $(\bm{0}, \mmstate{d}, \sigma)$. They are both right triangles with right angle located at point $\mmstate{d}$, since 
${\mmstate{d} \in (\mathrm{Dens}(d) - \mmstate{d})^\perp}$ --- that is, for any $S \in \mathrm{Dens}(d)$ one has $ \langle \mmstate{d},\ S - \mmstate{d} \rangle = 0$, because:
\begin{equation}
    \langle \mmstate{d},\ S - \mmstate{d} \rangle = \langle \mmstate{d}, S \rangle - \langle \mmstate{d}, \mmstate{d} \rangle := \Tr[(\mmstate{d}) S] - \Tr[(\mmstate{d})^2] = \Tr[S]/d - 1/d = 0\,.
\end{equation}
\end{proof}
\end{appendix}

\end{document}